\documentclass{article}
\usepackage{amsmath}  
\usepackage{amsthm}
\usepackage{arxiv}

\usepackage[utf8]{inputenc} 
\usepackage[T1]{fontenc}    
\usepackage{hyperref}       
\usepackage{url}            
\usepackage{booktabs}       
\usepackage{amsfonts}       
\usepackage{nicefrac}       
\usepackage{microtype}      
\usepackage{cleveref}       
\usepackage{lipsum}         
\usepackage{graphicx}
\usepackage{natbib}
\usepackage{doi}

\usepackage{subfigure}

\usepackage{algorithm}
\usepackage{algpseudocode}

\newtheorem{definition}{Definition}
\newtheorem{theorem}{Theorem}
\newtheorem{lemma}{Lemma}

\newcommand{\norm}[1]{\left\| #1 \right\|}
\makeatletter
\def\namedlabel#1#2{\begingroup
	#2%
	\def\@currentlabel{#2}%
	\phantomsection\label{#1}\endgroup
}

\algnewcommand{\Inputs}[1]{%
	\State \textbf{Inputs:}
	\Statex \hspace*{\algorithmicindent}\parbox[t]{.8\linewidth}{\raggedright #1}
}
\algnewcommand{\Initialize}[1]{%
	\State \textbf{Initialize:}
	\Statex \hspace*{\algorithmicindent}\parbox[t]{.8\linewidth}{\raggedright #1}
}

\title{$\alpha$-Rank-Collections: Analyzing Expected Strategic Behavior with Uncertain Utilities}

\date{}

\author{ \href{https://orcid.org/0000-0002-5712-1706}{\includegraphics[scale=0.06]{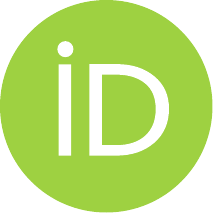}\hspace{1mm}Fabian R.~Pieroth} \\
	Department of Computer Science\\
	Technical University of Munich\\
	\texttt{fabian.pieroth@tum.de} \\
	\And
	\href{https://orcid.org/0000-0001-5491-2935}{\includegraphics[scale=0.06]{orcid.pdf}\hspace{1mm}Martin ~Bichler} \\
	Department of Computer Science\\
	Technical University of Munich\\
	\texttt{bichler@cit.tum.de} \\
}

\renewcommand{\headeright}{}
\renewcommand{\undertitle}{}
\renewcommand{\shorttitle}{$\alpha$-Rank-Collections}

\hypersetup{
pdftitle={\texorpdfstring{$\alpha$}{alpha}-Rank-Collections: Analyzing Expected Strategic Behavior with Uncertain Utilities},
pdfsubject={cs.GT},
pdfauthor={Fabian R.~Pieroth, Martin ~Bichler},
pdfkeywords={\texorpdfstring{$\alpha$}{alpha}-Rank; computational economics; decision analysis; utility theory},
}

\begin{document}
\maketitle

\begin{abstract}
Game theory relies heavily on the availability of cardinal utility functions, but in fields such as matching markets, only ordinal preferences are typically elicited. The literature focuses on mechanisms with simple dominant strategies, but many real-world applications lack dominant strategies, making the intensity of preferences between outcomes important for determining strategies. Even though precise information about cardinal utilities is not available, some data about the likelihood of utility functions is often accessible.
We propose to use Bayesian games to formalize uncertainty about the decision-makers' utilities by viewing them as a collection of normal-form games. Instead of searching for the Bayes-Nash equilibrium, we study how uncertainty in utilities is reflected in uncertainty of strategic play. To do this, we introduce a novel solution concept called $\alpha$-Rank-collections, which extends $\alpha$-Rank to Bayesian games. This allows us to analyze strategic play in, for example, non-strategyproof matching markets, for which appropriate solution concepts are currently lacking.
$\alpha$-Rank-collections characterize the expected probability of encountering a certain strategy profile under replicator dynamics in the long run, rather than predicting a specific equilibrium strategy profile. We experimentally evaluate $\alpha$-Rank-collections using instances of the Boston mechanism, finding that our solution concept provides more nuanced predictions compared to Bayes-Nash equilibria. Additionally, we prove that $\alpha$-Rank-collections are invariant to positive affine transformations, a standard property for a solution concept, and are efficient to approximate.
\end{abstract}


\section{Introduction}
Utility functions are used to analyze preferences and decisions analytically. However, different types of utility functions are used in the game-theoretical analysis. 
Cardinal utility functions assume that the utilities obtained for different outcomes are real-valued. The differences between the utilities of two outcomes are meaningful but not necessarily comparable across individuals. This plays an important role when analyzing decisions under uncertainty. The von Neumann-Morgenstern (vNM) utility function allows a mapping of outcomes to a cardinal utility, and it is commonly assumed in game theory \citep{vonneumannTheoryGamesEconomic1944}. 

A vNM function specifies the preferences of an individual precisely, but it is usually difficult to elicit \citep{hauserAssessmentAttributeImportances1979}. 
However, due to the precise information about a decision-maker's preferences, knowledge of the vNM functions allow for precise solution concepts as well -- the prevalent one being the Nash equilibrium (NE) \citep{nash1950equilibrium}. A NE is a fixed point in strategy space that no player unilaterally wants to deviate from, giving a single (possibly mixed) strategy profile as the solution.

Some streams of literature do require not only cardinal utility but also interpersonal utility comparisons. For example, the Bergson-Samuelson social welfare function is based on utilities of individuals that are comparable \citep{burkReformulationCertainAspects1938}. Likewise, auction theory and large parts of the mechanism design literature assume vNM utility functions and interpersonal utility comparisons. These strong assumptions allow for landmark results such as the dominant-strategy incentive-compatible Vickrey-Clarke-Groves mechanism \citep{nisanAlgorithmicGameTheory2007}. Notably, the Bayes-Nash equilibrium (a generalization of NE for games with incomplete information) provides a solution concept for finding equilibrium in mechanisms without dominant strategies.
For many applications such as spectrum sales or electricity markets, cardinal utilities, and even interpersonal (or inter-firm) utility, comparisons in money might be adequate. However, for others, they are too much to ask for. 

For example, matching with preferences typically assumes that only ordinal utilities can be elicited. 
Ordinal utilities provide only the rankings of utilities received from different bundles of goods or choices. It does not require individuals to specify how much extra utility they received from the preferred bundle of goods or services in comparison to other bundles. Decision-makers are only required to tell which bundles they prefer. While cardinal preferences can be challenged, eliciting only ordinal preferences is much less controversial in applications such as school choice or kidney exchanges. The Gale-Shapley algorithm \citep{galeCollegeAdmissionsStability1962} is arguably the most well-known example, which takes ordinal preferences as an input and computes a stable outcome. Importantly, the mechanism is strategy-proof for one side of the market. The Top Trading Cycles algorithm \citep{shapleyCoresIndivisibility1974} for the housing market enjoys the same property, but also here, preferences are restricted to unit demand. 

Strategy-proof mechanisms are rare, and most mechanisms used in the field do not possess this desirable property. For example, the Boston mechanism is widely used in school choice applications, but we lack a principled approach to studying strategic behavior in mechanisms that are not strategy-proof. 
Unfortunately, Nash equilibrium strategies for ordinal games (such as matching with preferences) are much less well understood. \citet{dasguptaExistenceEquilibriumDiscontinuous1986} were the first to analyze equilibrium in games with discontinuous preferences. \citet{renyNashEquilibriumDiscontinuous2016} and \citet{carmonaExistenceNashEquilibrium2016} discussed existence of Nash equilibria in ordinal games with ordered preferences. See \citet{renyNashEquilibriumDiscontinuous2020} for an up-to-date survey. If an ordinal Nash equilibrium exists, it is a NE for every possible vNM utility function that follows the ordinal preferences of the players. This is a strong requirement that is often not satisfied, and in many ordinal games, ordinal Nash equilibrium does not exist.
Importantly, we lose essential information about the intensity of preferences. However, in applications such as school choice or kidney exchanges, some vNM utility functions are more likely than others. 

In this paper, we study games where we have probabilistic information about the vNM functions of players. For example, in school choice, it might be clear that a student prefers one school much more than another, but it is not exactly clear by how much. We model this probabilistic information about the vNM functions of opponents via types in a Bayesian game. However, we do not search for a Bayes-Nash equilibrium in these games, which share the same problems as standard Nash equilibrium: there can be many equilibria, and finding equilibrium is PPAD-hard \citep{daskalakisComplexityComputingNash2009}.
Instead, we follow a novel approach that is fundamentally different from the NE's normative one: Markov-Conley-Chains (MCCs) and the $\alpha$-Rank algorithm, which were recently proposed for normal-form games \citep{omidshafieiARankMultiAgentEvaluation2019}. 

MCCs describe a solution concept extending the notion of a Nash equilibrium: for a strict Nash equilibrium, there exists an MCC, which is a singleton. Otherwise, the MCCs comprise a larger part of the strategy space describing the dynamics resulting from the underlying vNM utility functions. It is important to note that MCCs are no equilibrium solution concept. Instead, they arise from a simple (and tractable) behavioral model of how players learn. The $\alpha$-Rank algorithm is based on MCCs, and it computes a distribution over the strategy space, representing the probabilities for each strategy that it will invade other existing strategies. This way, it ranks the different strategies and avoids a selection problem inherent in the Nash equilibrium concept. Moreover, the algorithm can be computed in polynomial time\footnote{See Section \ref{subsec:computatial-costs} for details.}, providing an attractive alternative to the Nash equilibrium concept. However, $\alpha$-Rank and MCCs were defined for games with each player's unique vNM utility function. 

Our contribution is threefold: first, we use Bayesian games to model probabilistic information about the vNM functions of other players\footnote{In ordinal mechanisms such as matching markets, these vNM functions obey ordinal preferences by the players. However, in contrast to auction theory, we do not need interpersonal utility comparisons.} as a collection of normal-form games.
Second, we extend $\alpha$-Rank to Bayesian games. This combination provides a principled way to study strategic interaction in games, where this was not possible so far. For example, we can analyze strategies in matching markets that are not strategy-proof. Finally, we use the popular Boston mechanism to illustrate the new method but also apply it to a version of the Hawk-Dove game, where there is uncertainty about other players' types. The study of the Boston mechanism draws on an experimental paper by \citet{featherstone2016boston}. It is particularly insightful because we do not have game-theoretical predictions unless we assume we know the vNM utility functions of all players exactly, which is unreasonable to believe. 
Thirdly, we provide two theoretical results that establish invariance to the choice of the vNM representation and stability towards small variations in the vNM functions itself. Overall, the paper extends the recent literature on the MCC solution concept to Bayesian games and provides a novel approach to analyze ordinal games.

\section{Related Literature}
The field of utility theory is concerned with answering what strategic entities are striving for \citep{fishburnUtilityTheory1968}. 
After a long history of thought, the vNM utility function has become standard for describing preferences in the presence of risky outcomes of different choices and in game theory. 
Von Neumann and Morgenstern showed that, if an individual's preferences satisfy certain axioms of rational behavior, there exists a utility function such that the individual acts as if to maximize this utility~\citep{vonneumannTheoryGamesEconomic1944}. The vNM utility function is a cardinal utility function, i.e., utility is measured on an interval scale \citep{handler1984measurement}.

However, well-known negative results demonstrate the difficulty of aggregating preferences to a social choice with only ordinal utility functions available. Arrow's general impossibility theorem is one of the central insights, showing that it is impossible to have a social welfare function that satisfies a certain set of apparently reasonable conditions \citep{arrowSocialChoiceIndividual1951}. Based on this impossibility, \citet{satterthwaiteStrategyproofnessArrowConditions1975} and \citet{gibbardManipulationVotingSchemes1973} showed that no (non-dictatorial) mechanism with more than two possible outcomes is dominant-strategy incentive-compatible (aka. strategy-proof). The limitations of deterministic mechanisms led to significant attention for probabilistic social choice more recently \citep{brandt2017rolling}.

While it can be difficult to elicit vNM utility functions precisely \citep{hauserAssessmentAttributeImportances1979}, an ordinal preference relation is easier to determine, yet any information about the intensity of preferences is lost.
In matching with preferences where the mechanisms only take ordinal preferences into account, participants typically have stronger preferences for some alternatives than for others. For example, a student might have a very strong preference to join an elite high school around the corner compared to any other school in her vicinity when it comes to school choice. While it might be difficult to describe her preference with a unique vNM function, it would well be possible to define a probability distribution over a collection of vNM utility functions. 
\citet{azizStableMatchingUncertain2020} and \citet{azizStableMatchingUncertain2022} study stable matchings under different models of uncertainty over players' ordinal preferences. This differs from our approach, as they focus on existence of stable matchings and mechanisms to compute them, whereas we are interested in players' strategic behavior under uncertainty over cardinal preferences. Instead, \citet{amanatidisPeekingOrdinalCurtain2021a} propose a different approach to mitigate the information loss by allowing a few cardinal queries about agents' preferences. We argue that our approach is strategically easier, as an oracle for cardinal utilities is usually not available and eliciting these directly from the agents distorts the strategic interaction. 

The ordinal Nash equilibrium (ONE) provides one possibility to analyze strategies in ordinal games such as matching markets \citep{cruzOrdinalGamesGeneralized2000, amorEquilibriaOrdinalGames2017, ehlersIncompleteInformationSingleton2007}. ONE only takes purely ordinal preferences into account. One must consider every possible vNM function that abides by the underlying preference structure. Not surprisingly, ordinal Nash equilibria rarely exist. 
In most applications, the case of agents having no intensity about their preferences between different choices is unlikely. But then agents' strategies cannot be studied without taking these intensities into account. 
\citet{harsanyiGamesIncompleteInformation1967} introduced the concept of Bayesian games that express incomplete information via different types with a corresponding prior distribution over the types. We use this concept to model uncertainty about players' preferences. However, we differ from Harsanyi in assuming that uncertainty over players' types persists until the ex-post stage.

This approach shifts the focus to determining what the expected strategic behavior for such a collection of games is, where players are unsure about others and their own utilities. Answering this question makes it essential to compare strategic behavior for different instances of normal-form games.

The standard solution concept for Bayesian games is the Bayes-Nash equilibrium (BNE), which generalizes the notion of a NE. However, the NE already suffers from several problems. 
\citet{daskalakisComplexityComputingNash2009} showed that finding NE is PPAD complete, and \citet{rubinstein2016settling} demonstrated that it is hard to approximate. Furthermore, a single game can have many different NE, leading to the equilibrium selection problem \citep{harsanyiSolutionConceptNperson1976}. These issues make the BNE or NE challenging.
\citet{seltenReexaminationPerfectnessConcept1975} studied the stability of NE under perturbations of the players' strategies and introduced the notion of trembling hand. One defines a sequence of perturbed games, where each player has to play all available strategies with at least a probability of $\epsilon_n$. With $\epsilon_n \rightarrow 0$, one observes the NE in the perturbed games. If a sequence of NE exists in the perturbed games that converge to a NE in the original game, this NE is said to be a trembling hand perfect NE. Unlike the NE, a trembling hand perfect NE is stable under small changes in the vNM values. However, it still suffers from intractability and multiplicity \citep{DBLP:conf/sagt/HansenMS10}. 

\citet{papadimitriouGameDynamicsMeaning2019} propose to follow a more holistic approach than single stationary points as game outcome. Instead, they focus on game dynamics emerging from learning, i.e., adapting the players' behavior due to new experiences. \citet{omidshafieiARankMultiAgentEvaluation2019} builds on this narrative and introduces the concept of Markov-Conley-Chains (MCCs), which combines learning behavior with notions of convergence and stability of dynamical systems. MCCs extend the notion of NE, e.g., for a strict NE, there exists an MCC, which is a singleton consisting of the strict NE. On the other hand, if the NE is not strict, then MCCs may consist of a larger part of the strategy space, incorporating the agents' cycling or possibly even chaotic behavior. 
In particular, \citet{omidshafieiARankMultiAgentEvaluation2019} propose $\alpha$-Rank, a new algorithm to rank strategies and MCCs, solving the selection problem for this solution concept. It uses the finite population replicator dynamics model with selection intensity $\alpha$ to create a stationary distribution over the strategy space. 
The distribution measures the long-time behavior of the learning dynamics and ranks the different strategies according to the basins of attraction in the evolutionary dynamics. Furthermore, with increasing $\alpha$, the Markov chain associated with the derived stationary distribution coincides with the MCC. 
$\alpha$-Rank is efficient in the sense that it finds the stationary distribution for a given game dynamics in polynomial time. This and its connection to the replicator dynamics makes it arguably a better predictor for systems with strategically interacting agents. We use $\alpha$-Rank at the heart of our newly proposed solution concept and show that it is continuously differentiable in changes in the vNM functions.

There have been several papers that build on the $\alpha$-Rank paper's results. 
\citet{rowlandMultiagentEvaluationIncomplete2019}, \citet{duEstimatingARankFew2021}, and \citet{rashidEstimatingARankMaximizing2021} considered incomplete information about the true pay-off matrix, i.e., one can only sample the matrix's entries. They propose methods to reduce the number of samples needed to get the correct $\alpha$-Ranking from the approximated pay-off matrix with high probability. Note that this is a conceptually different approach to ours, as we explicitly want to capture the variability in the player's pay-offs while they see incomplete information as noisy measurements.
The joint strategy space grows exponentially with the number of agents, which results in $\alpha$-Rank only being applicable to moderate games. 
\citet{yangAlphaAlphaRank2020} use an oracle mechanism in combination with a stochastic solver, bypassing the need to build the full pay-off matrix before estimating the $\alpha$-Rank distribution. Especially for larger games, determining our newly proposed solution concept relies on approximation methods. Therefore, a more efficient computation of $\alpha$-Rank distributions is crucial to extending our method to larger domains.

\section{Preliminaries}
Let us first discuss the properties of ordinal and cardinal utilities before we introduce $\alpha$-Rank.

\subsection{Ordinal and Cardinal Utility Theory} \label{subsec:ordinal-and-cardinal-utilities}

In order to introduce central notions in utility theory, we revisit the well-known Hawk-Dove game \citep{kimStatusSignalingGames1995}. Two entities are competing over an indivisible resource that they assign a value of $V_1$ and $V_2$, respectively. They can choose between two actions. Either they choose to be a Hawk and fight over the resource or act as a Dove and retreat. If both choose Hawk, there will be a war where they win or lose with a $50\%$ probability. Winning results in getting the resource while losing costs $C_1$ and $C_2$, respectively. If one chooses Hawk while the other decides to play Dove, the Hawk player receives the resource while the other receives nothing $N_i$ for $i$ playing Dove. If both choose to play Dove, each player has a chance of $50\%$ of getting the resource or receiving nothing. The payoff matrix of the general Hawk-Dove game is depicted in Table \ref{tab:general-hawk-dove-pay-off-matrix}. In this interaction, each player has a set of different outcomes, where receiving the resource $V_i$ is preferred over getting nothing $N_i$, while this is better than paying the cost of war $C_i$. We say that player $i$ has the ordinal preference relation $\succeq_{i}^{\text{ord}}$ over the outcomes $\mathcal{O}_i=\{V_i, N_i, C_i\}$ with $V_i \succeq_{i}^{\text{ord}} N_i \succeq_{i}^{\text{ord}} C_i$. 
\begin{table}[ht]
	\centering
	\caption{The general Hawk-Dove interaction's pay-off matrix.}
	\label{tab:general-hawk-dove-pay-off-matrix}
	\begin{tabular}{l|c|c}
		&
		Hawk &
		Dove \\ \hline
		Hawk &
		\begin{tabular}[c]{@{}c@{}}$\frac{1}{2}V_1 + \frac{1}{2}C_1$, \\ \\ $\frac{1}{2}V_2 + \frac{1}{2}C_2$\end{tabular} &
		$V_1, N_2$ \\ \hline
		Dove &
		$N_1, V_2$ &
		\begin{tabular}[c]{@{}c@{}}$ \frac{1}{2}V_1 + \frac{1}{2}N_1$,\\ \\ $\frac{1}{2} V_2 + \frac{1}{2}N_2 $\end{tabular}
	\end{tabular}
\end{table}
Even though ordinal relations are available, one cannot analyze this interaction as an ordinal game. One would need to have knowledge about a ranking of the probabilistic outcomes $\frac{1}{2} V_i + \frac{1}{2} C_i$ and $\frac{1}{2} V_i + \frac{1}{2} N_i$. For example, does player 1 value the outcome $\frac{1}{2} V_1 + \frac{1}{2}C_1$ higher or lower than receiving nothing. This depends on \emph{how much better} he values $V_1$ to $N_1$ and \emph{how much worse} he ranks $C_1$ compared to $N_1$. The preference relation $\succeq_{i}^{\text{ord}}$ does not provide this information. One could argue that one needs to add the cases $\frac{1}{2} V_i + \frac{1}{2} C_i$ and $\frac{1}{2} V_i + \frac{1}{2} N_i$ to the set of outcomes, leading to an ordinal game. However, there could be only a few outcomes, but many different lotteries over these outcomes, and each additional ranking increases measurement difficulty. Furthermore, adding a distribution over outcomes as an outcome by itself and incorporating it into the ordinal preference relation $\succeq_{i}^{\text{ord}}$ restricts the interaction. 
There is no principled way to analyze interaction with only ordinal utilities over outcomes where decisions do not lead to a deterministic outcome.

Cardinal utilities in the form of a vNM function allow us to resolve these problems. 
Suppose an entity is faced with making a decision that can result in different outcomes $\mathcal{O}=\{o_1, \dots, o_K\}$. However, it is not deterministic which action results in which outcome. Instead, different actions can result in different lotteries over outcomes $L \in \Delta(\mathcal{O})$, where $\Delta(\mathcal{O})$ denotes the set of probability distributions over $\mathcal{O}$. If one wants to decide or predict which action the individual is going to make, one needs to know which action results in the most preferred lottery over the outcomes. These lotteries over outcomes are ranked by the individual's inherent preference relation $\succeq^{\text{vNM}}$. \citet{vonneumannTheoryGamesEconomic1944} showed that there exists a function $U: \mathcal{O} \rightarrow \mathbb{R}$ assigning a number to each outcome that is consistent with $\succeq^{\text{vNM}}$, if it satisfies certain rationality axioms \footnote{For a more detailed discussion of the axioms, see Chapter 5 of \citep{bonanno2018game}}. The individual then chooses the action that yields the highest expected utility.
The knowledge of vNM functions for a given interaction leads to the most basic formulation of strategic interaction, namely normal-form games.

\begin{definition}
	A normal-form game with $N$ players is described by a Tuple $(\mathcal{N}, \mathcal{S}, u)$.
	\begin{itemize}
		\item $\mathcal{N} = \{1, \dots N\}$ is a set of players.
		\item $\mathcal{S} = \mathcal{S}_1 \times \mathcal{S}_2 \times \cdots \times \mathcal{S}_N$ is the set of strategy profiles, where $\mathcal{S}_i$ is a finite set of strategies available to player $i \in \mathcal{N}$. A strategy profile $s \in \mathcal{S}$ is denoted by $s = (s_i, s_{-i})$, where the index $-i$ denotes the set $\mathcal{N} \setminus {i}$.
		\item $u = (u_1, \dots, u_N)$ is a tuple of utility functions, where $u_i : \mathcal{S} \rightarrow \mathbb{R}$ is player $i$'s utility function. 
	\end{itemize}
\end{definition}

A normal-form game implicitly incorporates vNM functions over outcomes in the following way. Each player $i$ has a set of possible outcomes $\mathcal{O}_i$ and a corresponding vNM function $U_i: \mathcal{O}_i \rightarrow \mathbb{R}$. Denote with $L_i(s) \in \Delta \mathcal{O}_i$ the lottery over player $i$'s outcomes given a strategy profile $s \in \mathcal{S}$. Then agent $i$'s utility function is given by $u_i(s) = \mathbb{E}_{o_i \sim L_i(s)}[U_i(o_i)]$.

Cardinal utility functions make much stronger assumptions, but very different outcomes can arise in equilibrium depending on these assumptions.
For example, consider the following three instances of the Hawk-Dove game, with different vNM functions shown in Tables \ref{tab:hawk-dove-prisoners-dilemma-instance}, \ref{tab:hawk-dove-battle-of-sexes-instance}, and \ref{tab:hawk-dove-mixed-pd-bos-instance}.

\begin{table}
	\begin{minipage}[b]{.27\textwidth}
		\centering
		\caption{Hawk-Dove as Prisoners' Dilemma with\\
			$U_1(V_1) = U_2(V_2) = 4$, $U_1(N_1) = U_2(N_2) = 0$, $U_1(C_1) = U_2(C_2) = -2$.}
		\label{tab:hawk-dove-prisoners-dilemma-instance}
		\begin{tabular}{l|c|c}
			& Hawk   & Dove   \\ \hline
			Hawk & $1, 1$ & $4, 0$ \\ \hline
			Dove & $0, 4$ & $2, 2$
		\end{tabular}
	\end{minipage}\qquad
	\begin{minipage}[b]{.27\textwidth}
		\centering
		\caption{Hawk-Dove as Anti-Coordination game with \\
			$U_1(V_1) = U_2(V_2) = 2$, $U_1(N_1) = U_2(N_2) = 0$, $U_1(C_1) = U_2(C_2) = -4$.}
		\label{tab:hawk-dove-battle-of-sexes-instance}
		\begin{tabular}{l|c|c}
			& Hawk   & Dove   \\ \hline
			Hawk & $-1, -1$ & $2, 0$ \\ \hline
			Dove & $0, 2$ & $1, 1$
		\end{tabular}
	\end{minipage}\qquad
	\begin{minipage}[b]{.28\textwidth}
		\centering
		\caption{Hawk-Dove as mixed interaction with \\
			$U_1(V_1) = 4$, $U_1(N_1) = 0$, $U_1(C_1) = -2$, $U_2(V_2) = 2$, $U_2(N_2) = 0$, $U_2(C_2) = -4$.}
		\label{tab:hawk-dove-mixed-pd-bos-instance}
		\begin{tabular}{l|c|c}
			& Hawk   & Dove   \\ \hline
			Hawk & $1, -1$ & $4, 0$ \\ \hline
			Dove & $0, 2$ & $2, 1$
		\end{tabular}
	\end{minipage}
\end{table}

Table \ref{tab:hawk-dove-prisoners-dilemma-instance} is an instance of the famous Prisoner's Dilemma, where (Hawk, Hawk) is the unique Nash Equilibrium. The second example in Table \ref{tab:hawk-dove-battle-of-sexes-instance} constitutes a Anti-Coordination interaction, where three different NE exist. These are (Hawk, Dove), (Dove, Hawk), and a mixed equilibrium. The final example shown in Table \ref{tab:hawk-dove-mixed-pd-bos-instance} consists of an interaction where player 1 has the same vNM function that resulted in the Prisoner's Dilemma, and Player 2's vNM function equals the one from the players of the Anti-Coordination interaction. Here, the unique NE is the pure strategy profile (Hawk, Dove). The general Hawk-Dove game and these three realizations with different vNM functions illustrate that ordinal utilities over outcomes may define an interaction semantically. However, the dynamics very much depend on the additional information provided by cardinal preferences.

\subsection{Markov-Conley-Chains and \texorpdfstring{$\alpha$}{alpha}-Rank} \label{subsec:alpha-rank}

As we discussed earlier, the Nash equilibrium has been challenged. \citet{omidshafieiARankMultiAgentEvaluation2019} recently proposed the alternative concept of Markov-Conley chains (MCCs). It is grounded in the dynamical systems perspective, where players follow natural learning dynamics \citep{papadimitriouGameDynamicsMeaning2019}. The resulting dynamical system by applying, e.g., the replicator dynamics to a game, decomposes into transient and chain components. The sink chain components, i.e., recurrent parts of the system that the dynamics stay in once reached, are viewed as outcomes of the game. 
An MCC is defined on the game's response graph, a directed graph where each strategy profile $s \in \mathcal{S}$ is a node. An edge exists if and only if there is only a single deviating player $i$ whose utility does not decrease when transitioning from strategy profile $s$ to $s^{\prime}$.  An MCC is a Markov-chain's stationary distribution within a sink strongly connected component of the response graph. Importantly, any sink chain component of the game dynamics contains at least one sink strongly connected component of the response graph and, therefore, also a corresponding MCC. 
This differs from the Nash equilibrium's normative approach, which tells players how to play. Instead, an MCC captures the long-term behavior of players involved in the interaction by eliminating transient strategy profiles while ranking non-transient profiles by their strength.

However, MCCs still suffer from multiplicity. So, to solve the selection problem, \citet{omidshafieiARankMultiAgentEvaluation2019} proposed $\alpha$-Rank, which provides a unique asymptotic evolutionary ranking of strategy profiles that captures the inherent transient and recurrent nature of the underlying dynamics. Additionally, it can be computed efficiently (see Section \ref{subsec:computatial-costs} for more details) in general-sum games and holds information about the whole strategy space.

$\alpha$-Rank constructs the game graph of a normal-form game, which is similar to the response graph defined above. However, a directed edge exists from $s$ to $s^{\prime} \in \mathcal{S}$ if and only if $s$ and $s^{\prime}$ differ in a single player $i$'s strategy, i.e., $s^{\prime} = (s^{\prime}_i, s_{-i})$, regardless of changes in the player's utility.
Subsequently, $\alpha$-Rank uses the finite population replicator dynamics model to infer a random walk on the game graph. Player $i$'s strategy is represented by a finite population where each individual can be one of $i$'s feasible strategies $\mathcal{S}_i$. In a joint strategy profile $s$, we assume each player's population to be homogeneously $s_i$. The transition probability $C_{s, s^{\prime}}$ along an edge from $s$ to $s^{\prime}$ describes the probability that a mutant in the population of the deviating player $i$'s strategy overtakes the population. Therefore, the mutant's strength and probability of success are determined by the players' utilities and the selection intensity parameter $\alpha$.
The latter is determined numerically, as we describe below. More specifically, the transition probability for an existing edge from $s$ to $s^{\prime}$ with deviating player $i$ is given by:
\begin{align}\label{equ:alpha-rank-transition-prob-deviating-player}
	C_{s, s^{\prime}}(\alpha, u) =
	\begin{cases}
		\eta \frac{1 - e^{- \alpha (u_i(s^{\prime}_i, s_{-i}) - u_i(s_i, s_{-i}))}}{1 - e^{- m\alpha (u_i(s^{\prime}_i, s_{-i}) - u_i(s_i, s_{-i}))}}, &\text{ if } u_i(s^{\prime}_i, s_{-i}) \neq u_i(s_i, s_{-i}) \\
		\frac{\eta}{m}, &\text{ otherwise.}
	\end{cases}
\end{align}
Here, $m \in \mathbb{N}$ denotes the population size, $\alpha > 0$ the selection intensity, and $\eta = \left(\sum_k |\mathcal{S}_k | - 1 \right)^{-1}$. If there is no edge between $s, s^{\prime} \in \mathcal{S}$, then the transition probability $C_{s, s^{\prime}}$ equals zero. The probability of a self-transition for a given strategy profile $s$ is then given by
\begin{align} \label{equ:alpha-rank-transition-prob-self-transition}
	C_{s, s}(\alpha, u) = 1 - \sum_{k \in \mathcal{N}, s^{\prime} = (s_{-i}, s_k), s_k \in \mathcal{S}_k \setminus {s_i}} C_{s, s^{\prime}}(\alpha, u).
\end{align}

For a higher value of $\alpha$, it becomes less likely that the population representing the current strategy gets invaded by a mutant representing a strategy resulting in a strictly smaller utility. Considering the limit $\alpha \rightarrow \infty$, non-improving strategies will never invade existing ones. For more details on the transition model, we refer to \citet{omidshafieiARankMultiAgentEvaluation2019}.
The resulting Markov-chain following the random walk is irreducible and aperiodic, which means that a unique stationary distribution, the so-called $\alpha$-Rank distribution, exists. The probabilities represent not the frequencies of the strategies being played but the probabilities for each strategy to invade other existing strategies. That means it measures the average time spent by the evolutionary process in each strategy. For a given $\alpha$ and utility functions $u$, we denote the mapping from a given interaction to the $\alpha$-Rank distribution by $f(\alpha, u) = \pi_{\alpha, u}$.

Increasing the selection intensity $\alpha$ connects the $\alpha$-Rank distribution to the solution concept of MCCs. More specifically, strategy profiles with a non-zero mass of the limiting distribution $\lim_{\alpha \rightarrow \infty} f(\alpha, u) = \pi_{u}^*$, that constitute a connected component of the game graph, coincide with the MCCs of the finite population replicator dynamics model.

Due to the connection of $\alpha$-Rank and MCCs, in practice, one tries to find the largest $\alpha$ such that the Markov chain is still connected (considering the finite precision of computer systems) and a stationary distribution $f(\alpha, u)= \pi_{\alpha, u}$ exists. We denote this maximum $\alpha$ by $\alpha_{u}^{\text{pre}}$. 
See Algorithm \ref{alg:steps-to-estimate-single-integrand} for details.

\begin{algorithm}
	\caption{Approximate an $\alpha$-Rank distribution for a single instance}\label{alg:steps-to-estimate-single-integrand}
	\begin{algorithmic}
		\Inputs{Set of agents $\mathcal{N}$, strategy set $\mathcal{S}$,\\
			drawn types $v \in \mathcal{V}$ with joint utility function $u(\cdot, v)$.
		}
		\Initialize{$\alpha \gets 10^{-5}$ (sufficiently small so that $f(\alpha, u(\cdot, v)$ exists) \\
			Determine transition matrix $C(\alpha, u(\cdot, v))$ by Equations \ref{equ:alpha-rank-transition-prob-deviating-player} and \ref{equ:alpha-rank-transition-prob-self-transition} \\
			Calculate stationary distribution $f(\alpha, u(\cdot, v)$ for $C(\alpha, u(\cdot, v))$}\\
		\Procedure{}{}
		\While{$f(\alpha, u(\cdot, v)$ exists}
		\State Increase: $\alpha \gets 2 \cdot \alpha$
		\State Update: $C(\alpha, u(\cdot, v)) \gets C(\alpha, u(\cdot, v))$, $f\left(\alpha, u(\cdot, v)\right) \gets f\left(\alpha, u(\cdot, v)\right)$
		\EndWhile
		\State \Return $\alpha$ as $\alpha_{u(\cdot, v)}^{\text{pre}}$ and $f(\alpha_{u(\cdot, v)}^{\text{pre}}, u(\cdot, v)) \approx g(u(\cdot, v))$.
		\EndProcedure
	\end{algorithmic}
\end{algorithm}

\section{\texorpdfstring{$\alpha$}{alpha}-Rank-Collections}
We present $\alpha$-Rank-Collections, a novel solution concept to capture uncertainty in the players' strategic behavior that arises from uncertainty in the players' preferences. We proceed with showing some desirable properties of $\alpha$-Rank-Collections, and finally discuss its computational costs.

\subsection{Capturing Uncertainty in Strategic Behavior}
We want to analyze agent strategies in games considering the probabilistic information available about vNM utilities of others. A Bayesian game provides a possibility to describe these uncertainties as probabilities over types \citep{harsanyiGamesIncompleteInformation1967}.

\begin{definition}\label{def:bayesian-game}
	A Bayesian game is specified by a tuple $(\mathcal{N}, \mathcal{S}, \mathcal{V}, F, u)$ where
	\begin{itemize}
		\item $\mathcal{N} = \{1, \dots N\}$ is a set of players
		\item $\mathcal{S} = \mathcal{S}_1 \times \mathcal{S}_2 \times \cdots \times \mathcal{S}_N$ is the set of strategy profiles, where $\mathcal{S}_i$ is a finite set of strategies available to player $i \in \mathcal{N}$
		\item $\mathcal{V}_i$ is of set of types for player $i$ with $\mathcal{V} = \mathcal{V}_1 \times \cdots \times \mathcal{V}_N$
		\item $F$ is a joint distribution over types
		\item $u_i: \mathcal{S} \times \mathcal{V} \rightarrow \mathbb{R}$ is player $i$'s utility function and $u = (u_1, \dots, u_N)$
	\end{itemize}
	If the utility function $u_i$ only depends on the types $\mathcal{V}_i$ of agent $i$ but not on the other players' types $\mathcal{V}_{-i}$ for all $i$, then the game is said to have private values. 
\end{definition}

The set of player $i$'s types $\mathcal{V}_i$ allows a parametrization of vNM functions over the interaction's outcomes $\mathcal{O}_i = \{o_{i, 1}, \dots, o_{i, K_i}\}$. This can be achieved by choosing $\mathcal{V}_i \subset \mathbb{R}^{K_i}$ and set a corresponding vNM function $U_i(o_{i, j}, v_i) = v_{i, j}$ for each $o_{i, j} \in \mathcal{O}_i$ and $v_i \in \mathcal{V}_i$. The prior distribution specifies the likelihood of a vNM function implied by $v_i$ to occur. In this way, we say that the possible set of vNM functions is fully specified, enabling the modeling of arbitrary preferences over outcomes. We present a scenario that uses this modeling in Section \ref{subsec:matching-experiments}.

In a Bayesian game's interim state, the players know their type (i.e., vNM function) but not those of the others. 
Although we use the same formalization, we interpret the prior distributions as uncertainty about participants' specific vNM function.
Given a draw from the prior, we obtain a normal-form game where strategic behavior can be predicted. The distribution over vNM functions dictates the likelihood of each normal-form game occurring, and predictions must be weighted according to the ex-ante uncertainty to provide an overall prediction. Thus, Bayesian games enable us to reason about strategies by considering a collection of possible normal-form games.

Given the problems with Bayes-Nash equilibria, such as their computational hardness, and multiplicity, we turn our attention to $\alpha$-Rank but extend the concept to Bayesian games via $\alpha$-Rank-collections. 



\begin{definition}[$\alpha$-Rank-collection] \label{def:alpha-rank-collection}
	Let $\text{BG} = (\mathcal{N}, \mathcal{S}, \mathcal{V}, F, u)$ be a Bayesian game.
	Let $g(u(\cdot, v)) = \pi_{u(\cdot, v)}^* \in \Delta \mathcal{S}$ be the limiting $\alpha$-Rank distribution for $u$ and $v \in \mathcal{V}$. Then we define $\text{BG}$'s $\alpha$-Rank-collection by
	\begin{align}
		\Lambda_{\text{BG}} = \mathbb{E}_{v \sim F} \left[g(u(\cdot, v)) \right].
	\end{align}
\end{definition}

For a single $v \in \mathcal{V}$, the limiting distribution $g(u(\cdot, v)) = \pi_{u(\cdot, v)}^*$ corresponds to a weighted ranking of the considered dynamics' MCCs (see Section \ref{subsec:alpha-rank}). That is, $\pi_{u(\cdot, v)}^*$ holds the information on how likely any given strategy profile is to encounter in the long run. The expectation over the prior $F$ weights this according to the likelihood that the players have $v \in \mathcal{V}$ as their types. 

We can view $\alpha$-Rank collections also in the following way. Imagine we have an oracle predictor $O: \mathcal{M} \rightarrow \Delta \mathcal{S}$ that maps from the space of utility matrices to the set of probability distributions over the joint strategy space $\mathcal{S}$. Furthermore, let $X: \Omega \rightarrow \mathcal{M}$ be a random variable that captures uncertainty in the players' utilities (preferences). Then the best predictor for the expected outcome is given by the expected value $\mathbb{E}\left[O(X) \right]$. In this sense, $\alpha$-Rank-collections are the best predictor over the players' types with $\alpha$-Rank as oracle. 

\subsection{Stability and Smoothness}

Let us now prove two useful properties of $\alpha$-Rank-collections. 
First, a vNM function for a given preference relation $\succeq^{\text{vNM}}$ is not unique, however, it is up to positive affine transformations. More specifically, two vNM functions $U, \tilde{U}$ represent the same preference relation if and only if there exist $a >0$ and $b \in \mathbb{R}$ such that $\tilde{U} = aU + b$. Due to this property, \citet{harsanyiGeneralTheoryEquilibrium1988} argue that any game-theoretic solution concept should be invariant to positive affine transformations of the players' utilities. The following theorem establishes this result for $\alpha$-Rank-collections if the utility functions are linear in the types. 

\begin{theorem} \label{thm:alpha-rank-collection-invariant-to-positive-affine-transformations}
	Let $\text{BG} = (\mathcal{N}, \mathcal{S}, \mathcal{V}, F, u)$ be a Bayesian game. Furthermore, let every utility function $u_i$ be linear in the types, then $\alpha$-Rank-collections are invariant to positive affine transformations of the type space $\mathcal{V}$. If $\text{BG}$ additionally has private values, $\alpha$-Rank-collections are invariant to positive affine transformations in each individual type space $\mathcal{V}_i$. 
\end{theorem}
We refer to Appendix~\ref{sec:proof-invariance-positive-affine-transformations} for the proof.
Note that the linearity assumption in the players' types does not hold in general. However, by applying a full specification of the players' vNM functions via the types, as described in Section~\ref{subsec:ordinal-and-cardinal-utilities}, this is automatically the case. For example, one can represent any instance of an ordinal game. Furthermore, it means the choice of vNM functions' representatives does not play a role for $\alpha$-Rank-collections. 

In practice, we do not have access to the limiting $\alpha$-Rank distribution $g(u(\cdot, v))$. Instead, we approximate it by $f(\alpha_{u(\cdot, v)}^{\text{pre}}, u(\cdot, v))$, i.e., by finding the largest $\alpha$ such that one can compute a stationary distribution algorithmically (see Section~\ref{subsec:alpha-rank}). That means we approximate the $\alpha$-Rank-collection by 
\begin{align} \label{equ:approx-alpha-rank-collection}
	\tilde{\Lambda}_{\text{BG}} = \mathbb{E}_{v \sim F} \left[f(\alpha_{u(\cdot, v)}^{\text{pre}}, u(\cdot, v)) \right],
\end{align}
where $\alpha_{u(\cdot, v)}^{\text{pre}}$ needs to be approximated for the specific types $v$. 
We show next that the mapping on $\alpha$-Rank distributions given by $f(\alpha, u(\cdot, v))$ is continuously differentiable in both arguments if the utilities are continuously differentiable in the types. For our purposes in modeling uncertainty in vNM functions, it means that the solution is not too sensitive to changes in the vNM functions, which is a desirable property because for a large set $\mathcal{V}$ one relies on approximation methods. The approximation technique's error bounds often rely on the integrand's variance, which is easier to bound for continuously differentiable functions. For example, variance reduction methods, such as quasi-Monte-Carlo integration, only have a meaningful bound for continuous integrands \citep{caflischMonteCarloQuasiMonte1998}.
Even though the limiting $\alpha$-Rank distribution $g(u(\cdot, v))$ is not necessarily even continuous, each evaluation of $f(\alpha, u(\cdot, v))$ gives an indication for a small neighborhood. The experiments in Section~\ref{subsec:matching-experiments} show a smooth transition of mass among the strategy profiles with changing utilities.

\begin{theorem}\label{thm:continuity-of-alpha-rank-in-vNM-and-alpha}
	Let $\mathcal{N}$ be a set of agents, $\mathcal{S}$ a strategy space, and $\mathcal{V}_i \subset \mathbb{R}^{m_i}$ for $i \in \mathcal{N}$ and $\mathcal{V} = \mathcal{V}_1 \times \cdots \times \mathcal{V}_N$. Furthermore, let the function $u_{i}: \mathcal{S} \times \mathcal{V} \rightarrow \mathbb{R}$ be continuously differentiable in its second argument for every $i \in \mathcal{N}$. The tuple $\left(\mathcal{N}, \mathcal{S}, u(\cdot, v) \right)$ constitutes a normal-form game for every $v \in \mathcal{V}$. Then, the mapping on the corresponding $\alpha$-Rank distribution $f(\alpha, u(\cdot, v)) = \pi_{\alpha, u(\cdot, v)}$ is continuously differentiable in $\alpha \in (0, \infty)$ and $v \in \mathcal{V}$.
\end{theorem}

We refer to Appendix~\ref{sec:proof-of-continuity-of-alpha-rank} for the proof. Note that the condition of the utility functions being continuously differentiable in the types holds again for a full specification of the vNM functions as described in Section~\ref{subsec:ordinal-and-cardinal-utilities}.

\subsection{Approximation and Computational Costs} \label{subsec:computatial-costs}

To determine $\tilde{\Lambda}_{\text{BG}}$, one needs to use some appropriate way to estimate the expectation over $\alpha$-Rank distributions. For a finite set of types, this can be done explicitly, as the average of finitely many $\alpha$-Rank distributions weighted according to the prior. Otherwise, Monte-Carlo approximation \citep{hammersleyMonteCarloMethods1964} constitutes an easy alternative for continuous type spaces with potentially more complex priors. Furthermore, due to Theorem~\ref{thm:continuity-of-alpha-rank-in-vNM-and-alpha}, we know that Monte-Carlo approximation works with reasonable bounds. See Algorithm~\ref{alg:monte-carlo-approximation-alpha-rank-collection} for a high-level procedure applying Monte-Carlo approximation to estimate an $\alpha$-Rank-collection.

For both methods, one needs to evaluate the limiting $\alpha$-Rank distribution for a given type, which demands to determine $\alpha_{u(\cdot, v)}^{\text{pre}}$ for a given $v \in \mathcal{V}$ in a first step. However, often the types $\mathcal{V}$ are relatively close together, so that it suffices to determine only a few different $\alpha_{u(\cdot, v)}^{\text{pre}}$ for certain areas in $\mathcal{V}$. Consequently, this all boils down to the computational cost of determining a single $\alpha$-Rank distribution. 

$\alpha$-Rank is efficient in the sense that it finds the stationary distribution for a given game dynamics in polynomial time, i.e., it is pseudo-polynomial. Table~\ref{tab:eval-runtime-alpha-rank} reports the runtime for random instances of different sizes. The final row corresponds to the size of our empirical experiments presented in Section~\ref{subsec:matching-experiments}, with joint strategy space of size $7776$ it takes about $6.5$ minutes to calculate a single $\alpha$-Rank distribution. However, note that this can be parallelized easily so that this can be distributed on multiple machines.
\begin{table}
	\centering
	\caption{Estimated runtime of calculating an $\alpha$-rank distribution for a setting with $5$ agents and growing action spaces.}
	\label{tab:eval-runtime-alpha-rank}
	\begin{tabular}{lrrr}
		\toprule
		$|\mathcal{A}_i|$ &  Num strategies $|\mathcal{S}|$ & Runtime $[s]$ \\
		\midrule
		2 &                    32 &  0.03 \\
		3 &                   243 &  0.59 \\
		4 &                  1024 & 5.30  \\
		5 &                  3125 &  42.79 \\
		6 &                  7776 & 397.22   \\
		\bottomrule
	\end{tabular}
\end{table}
Nonetheless, the game graph itself grows exponentially with the number of agents and strategies, meaning that the originally proposed computation of $\alpha$-Rank is only feasible for relatively small settings. Furthermore, \citet{yangAlphaAlphaRank2020} conjecture that determining the top-rank joint strategy profile in $\alpha$-Rank is NP-hard. Given that our method necessitates computing multiple $\alpha$-Rank distributions, this potential computational complexity would also apply to ascertaining precise $\alpha$-Rank-Collections. However, since $\alpha$-Rank-Collections account for the entire $\alpha$-Rank distributions rather than solely the top-rank and represent an expectation over the prior distribution, our approach still offers valuable insights with approximate results. \citet{yangAlphaAlphaRank2020} introduce an algorithm to efficiently approximate the $\alpha$-Rank distribution using a double oracle mechanism, thereby eliminating the need to store the game graph's transition matrix in memory. Exploring larger settings and the computational complexity of approximation methods remains a promising direction for future research, albeit orthogonal to the conceptual introduction of $\alpha$-Rank-Collections.

\begin{algorithm}
	\caption{Monte-Carlo Approximation for $\alpha$-Rank-collections}\label{alg:monte-carlo-approximation-alpha-rank-collection}
	\begin{algorithmic}
		\Inputs{
			Bayesian game $\text{BG} = (\mathcal{N}, \mathcal{S}, \mathcal{V}, F, u)$\\
			Sampling number $n_{\text{sample}}$
		}
		\Initialize{
			$\alpha$-Rank distribution list $Q$
		}
		\Procedure{}{}
		\While{$\text{iter} < n_{\text{sample}}$}
		\State Sample $v \sim F$.
		\State Determine $\alpha_{u(\cdot, v)}^{\text{pre}}$ and $f(\alpha_{u(\cdot, v)}^{\text{pre}}, u(\cdot, v))$, e.g. via a sweep over values for $\alpha$
		\State Append $f(\alpha_{u(\cdot, v)}^{\text{pre}}, u(\cdot, v)) \approx g(u(\cdot, v))$ to $Q$.
		\State $\text{iter} \gets \text{iter} + 1$
		\EndWhile
		\State \Return $\text{MEAN}(Q) = \tilde{\Lambda}_{\text{BG}} \approx \Lambda_{\text{BG}}$.
		\EndProcedure
	\end{algorithmic}
\end{algorithm}

\section{Experimental Evaluation} \label{sec:empirical-evaluations}

This section illustrates $\alpha$-Rank-collections with two examples. The first is the previously introduced general Hawk-Dove game. Subsequently, we present a more involved experiment of a two-sided matching market. 

\subsection{The General Hawk-Dove Interaction}

In our Hawk-Dove game, the players' types represent the vNM functions over the outcomes of receiving the good $V_i$, receiving nothing $N_i$, or paying the cost of war $C_i$. We assume the players either to have a Prisoner's Dilemma [PD] type (see Table~\ref{tab:hawk-dove-prisoners-dilemma-instance}) or a Anti-Coordination [AC] type (see Table~\ref{tab:hawk-dove-battle-of-sexes-instance}). The prior $F$ specifies an agent to be a PD type with probability $p$ and an AC type with $1-p$. This can lead to four different types of interactions, namely [PD, PD], [PD, AC], [AC, PD], [AC, AC]. We conduct a sweep over increasing values of $\alpha$ for each instance to select the maximum $\alpha$ such that the $\alpha$-Rank distribution still exists. The sweeps' results are depicted in Figure~\ref{fig:alpha-sweep-hawk-dove-game}.

\begin{figure*}
	\centering
	\subfigure[$\text{[PD, PD]}$ interaction (Table \ref{tab:hawk-dove-prisoners-dilemma-instance})]{\label{fig:alpha-sweep-pd}
		\includegraphics[width=0.3\textwidth]{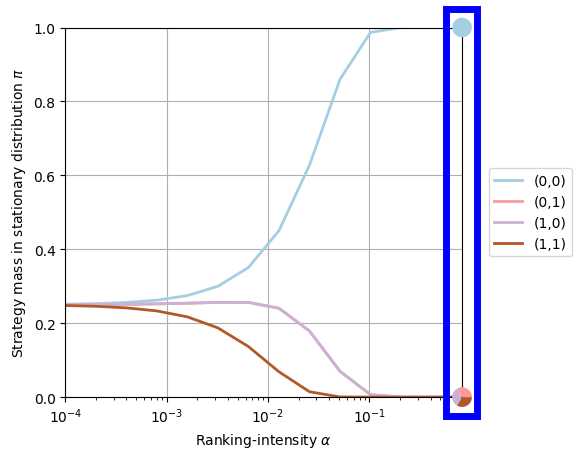}}
	\hfill
	\subfigure[$\text{[AC, AC]}$ interaction (Table \ref{tab:hawk-dove-battle-of-sexes-instance})]{\label{fig:alpha-sweep-bos}
		\includegraphics[width=0.3\textwidth]{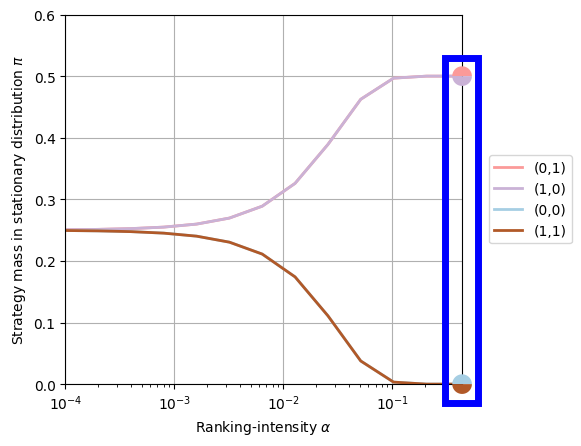}}
	\hfill
	\subfigure[$\text{[PD, AC]}$ interaction (Table \ref{tab:hawk-dove-mixed-pd-bos-instance})]{\label{fig:alpha-sweep-pd-bos}
		\includegraphics[width=0.3\textwidth]{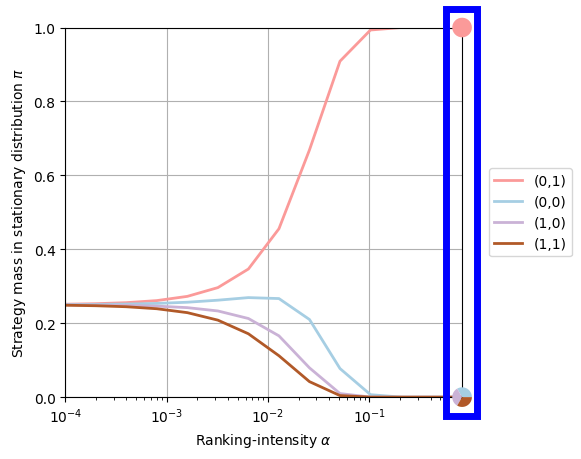}}
	\caption{Values for stationary distributions for different values of $\alpha$ over the strategy space. Action "0" corresponds to "Hawk" and action "1" to "Dove".}
	\label{fig:alpha-sweep-hawk-dove-game}
\end{figure*}

In the [PD, PD] case, the whole mass lies on (Hawk, Hawk), the unique NE. The converged stationary distribution for the Anti-Coordination instance puts equal mass on the strategies (Hawk, Dove) and (Dove, Hawk). Again, both are NE and weighted equally likely to occur in the long run. The mixed instances result in putting all mass on either the (Hawk, Dove)-strategy for the [PD, AC] case or on the (Dove, Hawk)-strategy for the [AC, PD] case. The expected value gives the $\alpha$-Rank-collection over the different $\alpha$-Rank distributions. 

\begin{align}\label{equ:alpha-rank-collection-general-hawk-dove-game}
	\Lambda_{\text{H-D}}(p) &= p^2 \cdot \pi_{\text{PD, PD}}^* + p(1-p) \cdot \pi_{\text{PD, AC}}^* \nonumber \\
	&+ (1-p)p \cdot \pi_{\text{AC, PD}}^* + (1-p)^2 \cdot \pi_{\text{AC, AC}}^* \nonumber \\
	&=
	\left(
	\begin{array}{c}
		\Lambda_{\text{H-D}}(p)[\text{Hawk, Hawk}]\\
		\Lambda_{\text{H-D}}(p)[\text{Hawk, Dove}]\\
		\Lambda_{\text{H-D}}(p)[\text{Dove, Hawk}]\\
		\Lambda_{\text{H-D}}(p)[\text{Dove, Dove}]
	\end{array}
	\right) =
	\left(
	\begin{array}{c}
		p^2\\
		\frac{1-p^2}{2}\\
		\frac{1-p^2}{2}\\
		0
	\end{array}
	\right)
\end{align}

The $\alpha$-Rank-collection $\Lambda_{\text{H-D}}(p)$ for a given parameter $p$ tells us the following. The strategy profile $(Dove, Dove)$ has an expected evolutionary strength of zero, which means that we expect this strategy profile to vanish. The strategy profiles $(Hawk, Dove)$ and $(Dove, Hawk)$ have the expected evolutionary strength of $\frac{1- p^2}{2}$, which is strictly above zero if $p < 1$. Finally, $(Hawk, Hawk)$ has an expected evolutionary strength of $p^2$. This result shows that even if all players almost certainly have a Prisoner's Dilemma type vNM function, for example, with $p=0.98$, one still expects the strategy profiles $(Hawk, Dove)$ and $(Dove, Hawk)$ to occur, as their combined expected evolutionary strength is close to $4$\%. Also, for $p< \sqrt{2}/2$, the evolutionary strength of $(Hawk, Hawk)$ is smaller than the combined of $(Hawk, Dove)$ and $(Dove, Hawk)$ and we expect $(Hawk, Hawk)$ to appear less frequently than the other two strategy profiles. 

\begin{figure}[ht]
	\centering
	\includegraphics[width=.6\textwidth]{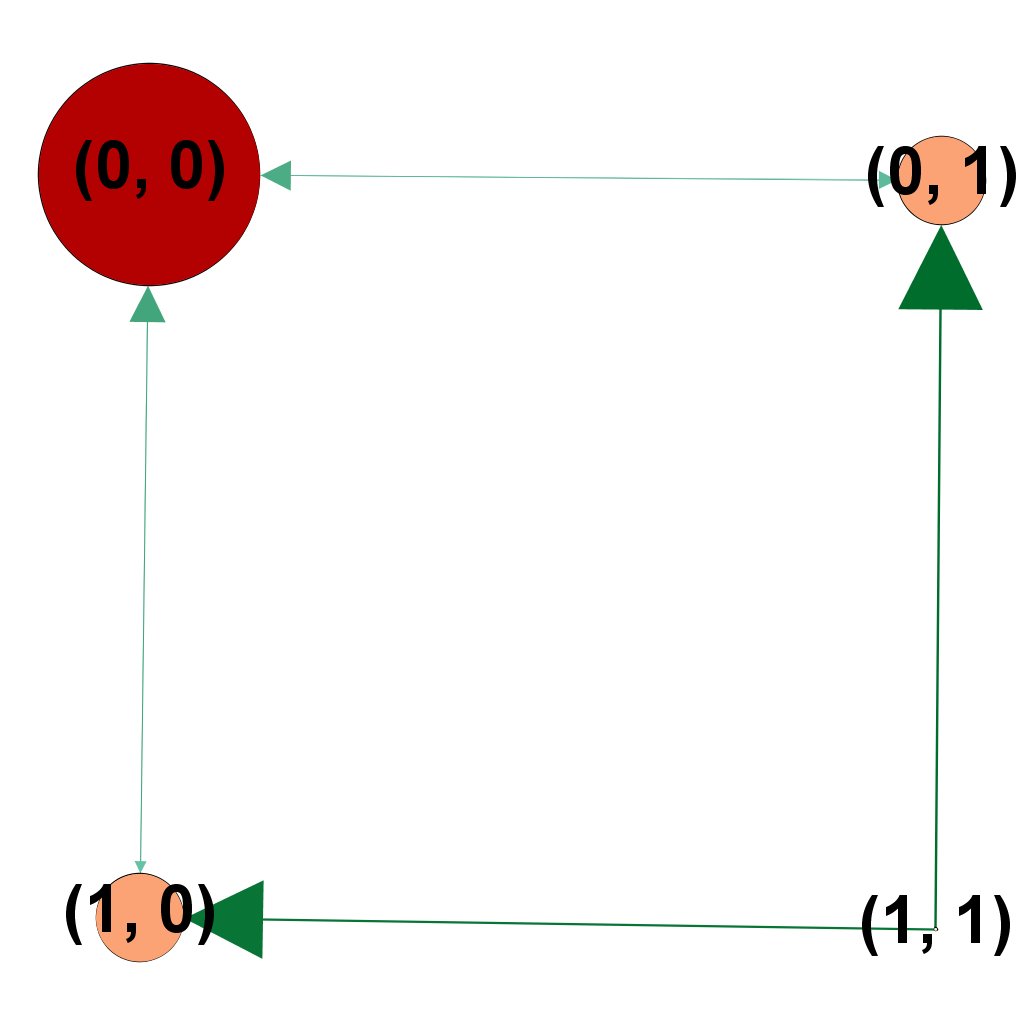}
	\caption{Game graph of the Hawk-Dove game graph with $p=0.75$. The nodes' sizes and colors represent the $\alpha$-Rank-collection's mass. Nodes with high mass are large and red, whereas low mass nodes are brighter and small. The edges' size represents the improvement in the deviating agent's utility weighted by the prior distribution over the four possible instances. Action "0" corresponds to "Hawk" and action "1" to "Dove".}
	\label{fig:graph-hawk-dove-collection-p-075}
\end{figure}

Figure~\ref{fig:graph-hawk-dove-collection-p-075} shows the averaged game graph\footnote{We use the software Gephi \citep{Bastian_Heymann_Jacomy_2009} for all graph visualizations.} for $p=0.75$, where the nodes' sizes and colors represent the $\alpha$-Rank-collection. One can see that (Hawk, Hawk) dominates the distribution. However, the edges pointing towards (Dove, Hawk) and (Hawk, Dove) show small expected improvements, so these strategies still appear in the long run.
Note that for $p \geq \frac{1}{2}$ the Bayesian Nash equilibrium of this game would be the pure strategy with playing Hawk as PD type and Dove as AC type. On the other hand, an equilibrium with distributional strategies exists only for the case of $p < \frac{1}{2}$, together with two pure strategy equilibria. This again highlights the NE's difference to $\alpha$-Rank-collections. Instead of having best-responses to strategies conditioned on the individual types that constitute a BNE, an $\alpha$-Rank-collection captures variances in strategic play resulting from uncertainties in the players' vNM functions.

\subsection{Two-Sided Matching Market} \label{subsec:matching-experiments}
The following section analyzes the Boston mechanism as an example for the game-theoretical analysis of an ordinal game. 

\subsubsection{Experimental Setting.} \label{subsubsec:experimental_setting}
We build upon the game setting introduced for the experimental setup considered by \citet{featherstone2016boston} in their laboratory study.
The \textit{aligned preference environment} is an instance of the school choice problem\footnote{Also known as the student-courses or hospital-residents problem.}. It consists of five students (3 Tops, 2 Averages) and three schools (Gold, Silver, Bronze). The Gold school has two seats, while the Silver and Bronze schools only have one seat, such that one student will remain unmatched. We denote the set of possible outcomes by $\mathcal{O}_i = \{G, S, B, U\}$, where $G$, $S$, $B$ stand for the Gold, Silver, and Bronze school respectively, and $U$ for being unmatched. All schools prefer the three Tops to the two Averages. Note that the students among these groups are ranked equally, and a lottery decides who gets a seat among these groups. For example, if the two Averages apply for the Silver school, a lottery decides who gets the seat. The schools are assumed to be not strategic, i.e., they always submit their true preferences. All students prefer the Gold to Silver to the Bronze school. All of these choices are preferred to being unmatched.
Therefore, we require students to submit a full preference ordering over the schools, i.e., $\mathcal{S}_i = \{(G, S, B), (G, B, S), (S, G, B), (S, B, G), (B, G, S), (B, S, G) \}$ for each student $i$. This environment is asymmetric because Tops and Averages are treated differently by the schools.

We study the environment under two different matching mechanisms. 
The first is the well-known student-optimal deferred acceptance\footnote{This is a variant of the Gale-Shapley matching mechanism.} (DA) mechanism \citep{galeCollegeAdmissionsStability1962}. The mechanism works as follows. First, the students apply to their most preferred school. The schools reject the lowest ranking students who exceed their capacity limit. Not rejected students are held \emph{temporarily}. If a student is rejected, she applies to the next school in the preference order. The schools take temporarily admitted and newly applying students into account and reject the lowest ranking students that exceed the capacity. This process repeats until no more students apply or are rejected. The DA mechanism is strategy-proof, i.e., truth-telling is a dominant strategy. 

The second mechanism is the Boston school matching mechanism \citep{abdulkadirogluBostonPublicSchool2005}. The mechanism works similarly to the DA mechanism. However, seats of admitted students are fixed in each round. All seats at a school may already be taken even if the applying student is ranked higher than any already submitted student. For example, all three Tops apply to the Gold school in the first round, whereas the Averages apply to the Silver and Bronze school. One Top remains unmatched and applies to her second choice. As all available seats are already taken under Boston, the Top remains unmatched, whereas she would get another seat under DA. 

We construct a normal-form game from this environment by choosing vNM functions for the students and calculate their expected utilities for all possible strategy profiles. We denote the cardinal values for each possible outcome by $v_i = (v_{G, i}, v_{S, i}, v_{B, i}, v_{U, i})$, where $v_i$ denotes player $i$'s type. For example, if all three Tops play truthfully under DA with resulting strategy profile $s^{\text{tr}}$, the lottery over the outcomes for a Top $i$ is given by $L_{i, \text{DA}}(s^{\text{tr}}) = \left(\frac{2}{3}, \frac{1}{3}, 0 , 0\right)$ with expected utility $u_{i, \text{DA}}(s^{\text{tr}}, v_i)) = \frac{2}{3} v_{G, i} + \frac{1}{3} v_{S, i}$. 

\subsubsection{Equilibrium Strategies.} \label{subsubsec:equilibrium-strategies-aligned-env}

The equilibrium strategies depend on the chosen mechanism. Under Boston, they are also dependent on the specific vNM functions, which is not the case under DA. 
The DA mechanism is strategy-proof, which means that truth-telling is a NE. However, the set of equilibrium strategies is given by all strategy profiles where the Tops play truthfully, and the Averages play an arbitrary strategy. These $36$ strategies denoted by $\text{NE}_{\text{DA}}$ are even an ordinal NE (ONE), i.e., they are a NE for every vNM function satisfying the ordinal preference relation (see Proposition 2.1 of \citep{featherstone2016boston}). We expect each of these strategy profiles to hold a significant amount of mass under the $\alpha$-Rank distribution, independent of the choice of vNM functions.

The Boston algorithm is not truthful, in fact, it is even easily manipulable \citep{pathakSchoolAdmissionsReform2011}.
\citet{featherstone2016boston} considered the case with students' vNM functions satisfying $\frac{2}{3}v_{G, i} + \frac{1}{3}v_{B, i} > v_{S, i}$, for every student $i$. In this case, the set of equilibrium strategies is given by the four strategies where Tops rank Gold first and Silver second, whereas the Averages rank Silver first (see Proposition 2.2 of \citep{featherstone2016boston}). We denote this set of equilibria by $\text{NE}_{\text{Bo}}$.

If there is at least one Top $j$ with a vNM function such that $\frac{2}{3}v_{G, j} + \frac{1}{3}v_{B, j} < v_{S, j}$, then none of the strategies from $\text{NE}_{\text{Bo}}$ is still an equilibrium strategy. That is due to this Top having a beneficial deviation by ranking Silver first. For the case $\frac{2}{3}v_{G, j} + \frac{1}{3}v_{B, j} = v_{S, j}$, the strategies $\text{NE}_{\text{Bo}}$ remain NE, however, these are no longer strict and do not constitute the whole set of equilibria. That means, NE under Boston depend on the players' types.

\subsubsection{Results.}
We start by analyzing the $\alpha$-Rank distributions for specific types of players. Afterward, we approximate an $\alpha$-Rank-collection for a Bayesian game formulation. First, we fix the vNM values for Gold and Bronze to $100$ and $25$, respectively, and vary the vNM value for Silver for all students. Figure~\ref{fig:mass_over_utility_landscapes} shows the mass put on the set of equilibrium strategies under DA $(\text{NE}_{\text{DA}})$ and Boston $(\text{NE}_{\text{Bo}})$. One can see that all mass remains on the set $\text{NE}_{\text{DA}}$ under DA. Furthermore, we observed that the mass is equally distributed among all $36$ different strategies. The behavior under Boston shows significant differences. For a vNM value of $70$ for Silver, about $97$\% of mass is concentrated on the strategies $\text{NE}_{\text{Bo}}$. For an increasing vNM value for Silver, the mass on the strategies $\text{NE}_{\text{Bo}}$ slightly decreases at the beginning and drops sharply around $75$. For values higher than $75$, no mass remains on the previous NE strategy profiles.

\begin{figure}[ht]
	\centering
	\includegraphics[width=.6\textwidth]{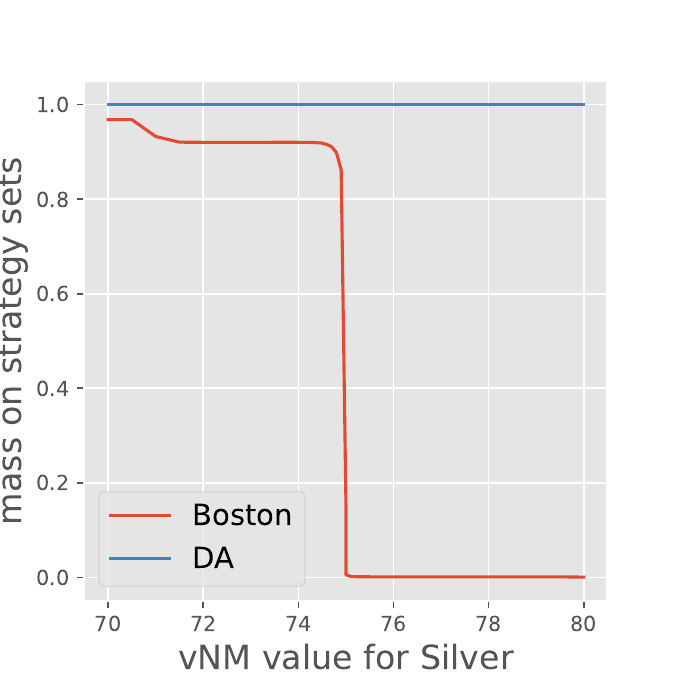}
	\caption{The $\alpha$-Rank distribution's mass on the strategies profiles $\text{NE}_{\text{DA}}$ under DA and $\text{NE}_{\text{Bo}}$ under Boston for different values of the vNM function for Silver. The values of Gold and Bronze are fixed to $100$ and $25$ respectively.}
	\label{fig:mass_over_utility_landscapes}
\end{figure}

Figure~\ref{fig:graph-vizs-boston} presents three different instances of game graphs where the nodes' masses are the $\alpha$-Rank distributions. Note that the set of strategy profiles $\mathcal{S}$ consists of $6$ strategies for each of the $5$ participants and therefore has $6^5=7776$ entries. As a result, the game graph is extensive. 
The graphs depict the strategy profiles as nodes, and a directed edge exists if a deviation of a single agent does not decrease its utility. The nodes' sizes and colors represent each strategy profiles mass under the $\alpha$-Rank distribution, representing the individual profiles' long-term probability of occurring.
The graphs show the masses' shift from the four strategies $\text{NE}_{\text{Bo}}$ towards a greater variety. Whereas visually, all mass is on the four NE strategies $\text{NE}_{\text{Bo}}$ in Figure~\ref{fig:graph-viz-standard-case-boston} with $v_S=70$, other strategies emerge on the right in Figure~\ref{fig:graph-viz-indifferent-case-boston} with $v_S=75$. In Figure~\ref{fig:graph-viz-strictly-case-boston}, with $v_S=80$, the $\alpha$-Rank distribution's mass shifted various strategies.

\begin{figure*}[ht]
	\centering
	\subfigure[$v_S=70$]{\label{fig:graph-viz-standard-case-boston}\includegraphics[width=0.3\textwidth]{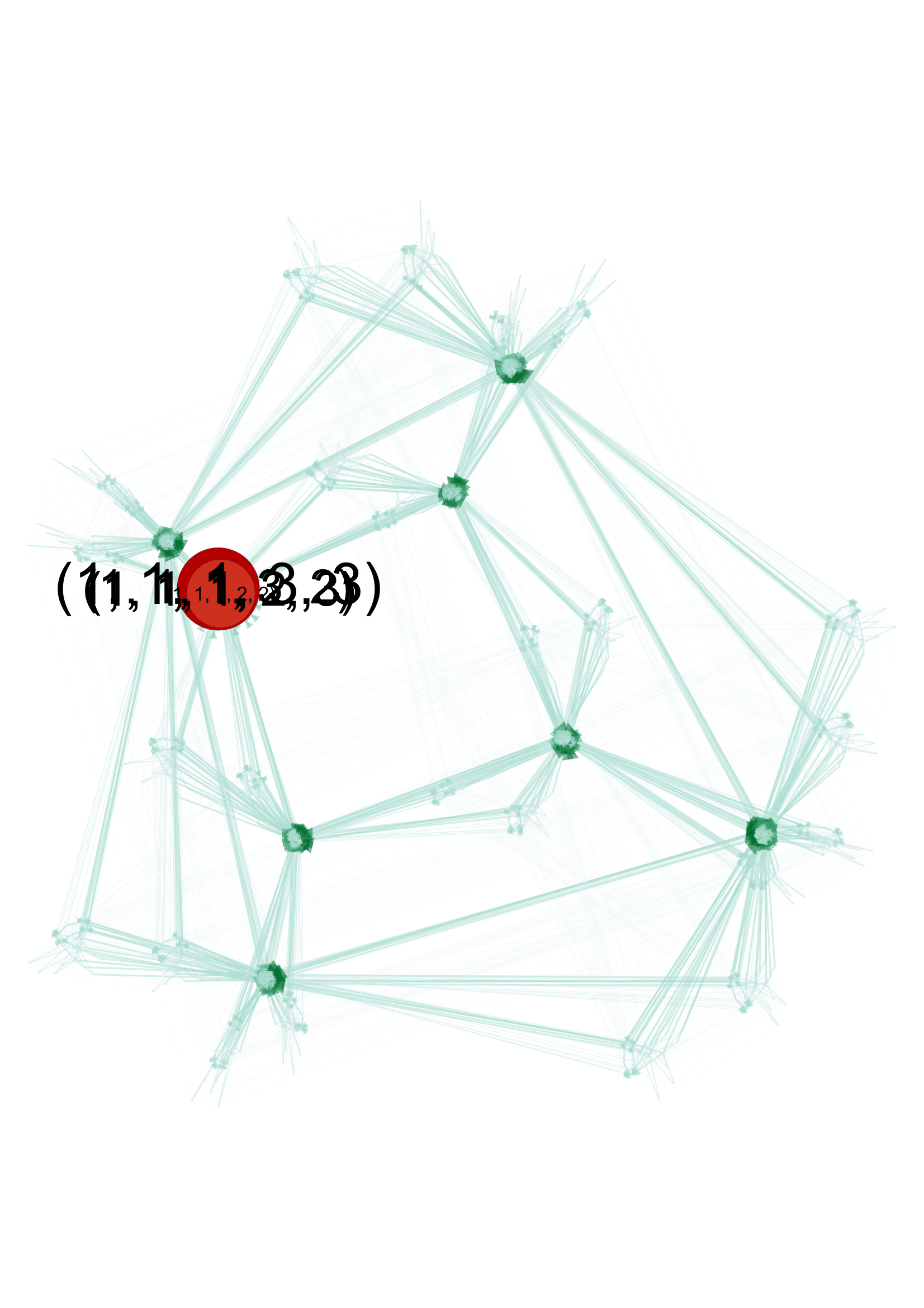}}
	\hfill
	\subfigure[$v_S=75$]{\label{fig:graph-viz-indifferent-case-boston}\includegraphics[width=0.3\textwidth]{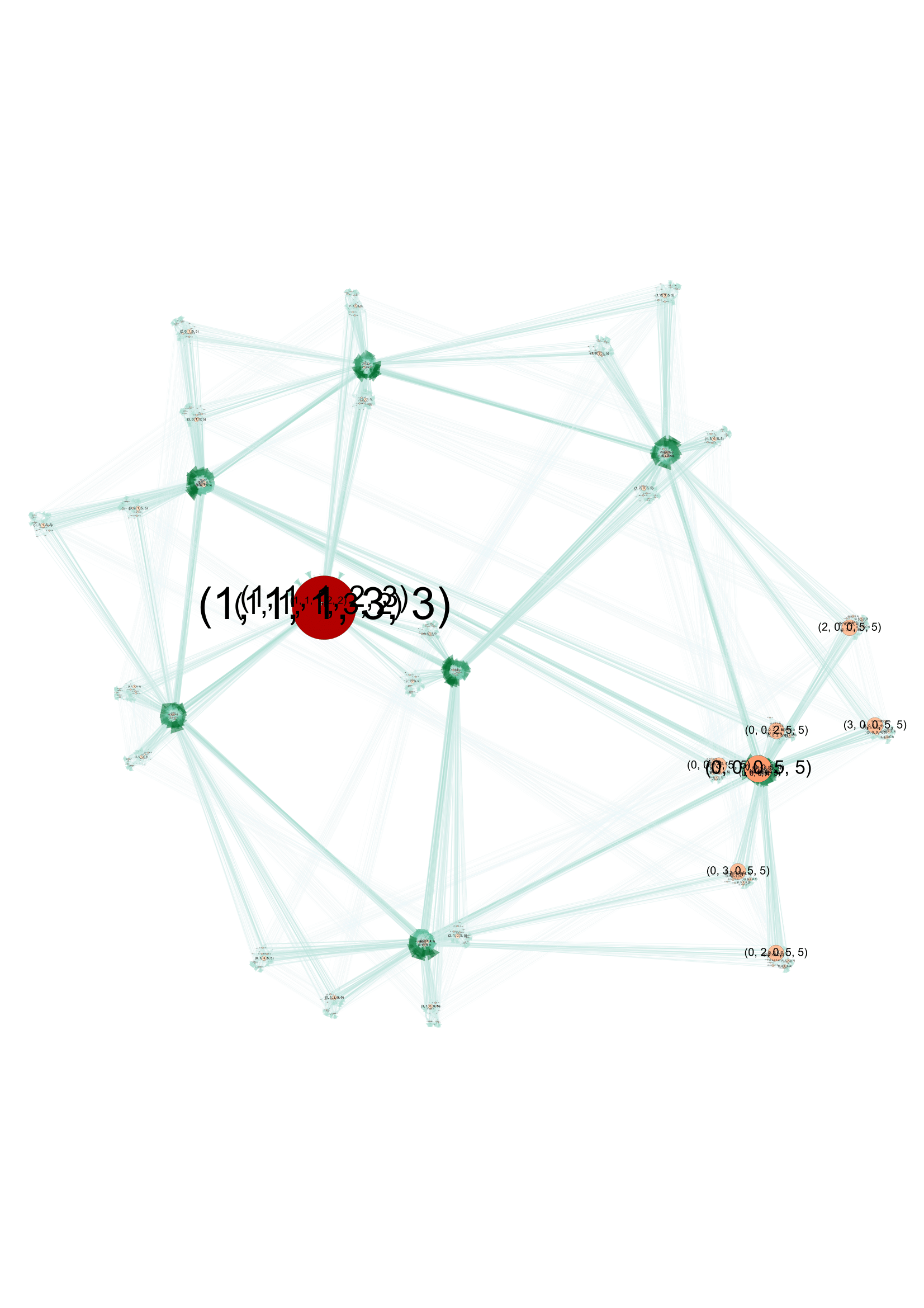}}
	\hfill
	\subfigure[$v_S=80$]{\label{fig:graph-viz-strictly-case-boston}\includegraphics[width=0.3\textwidth]{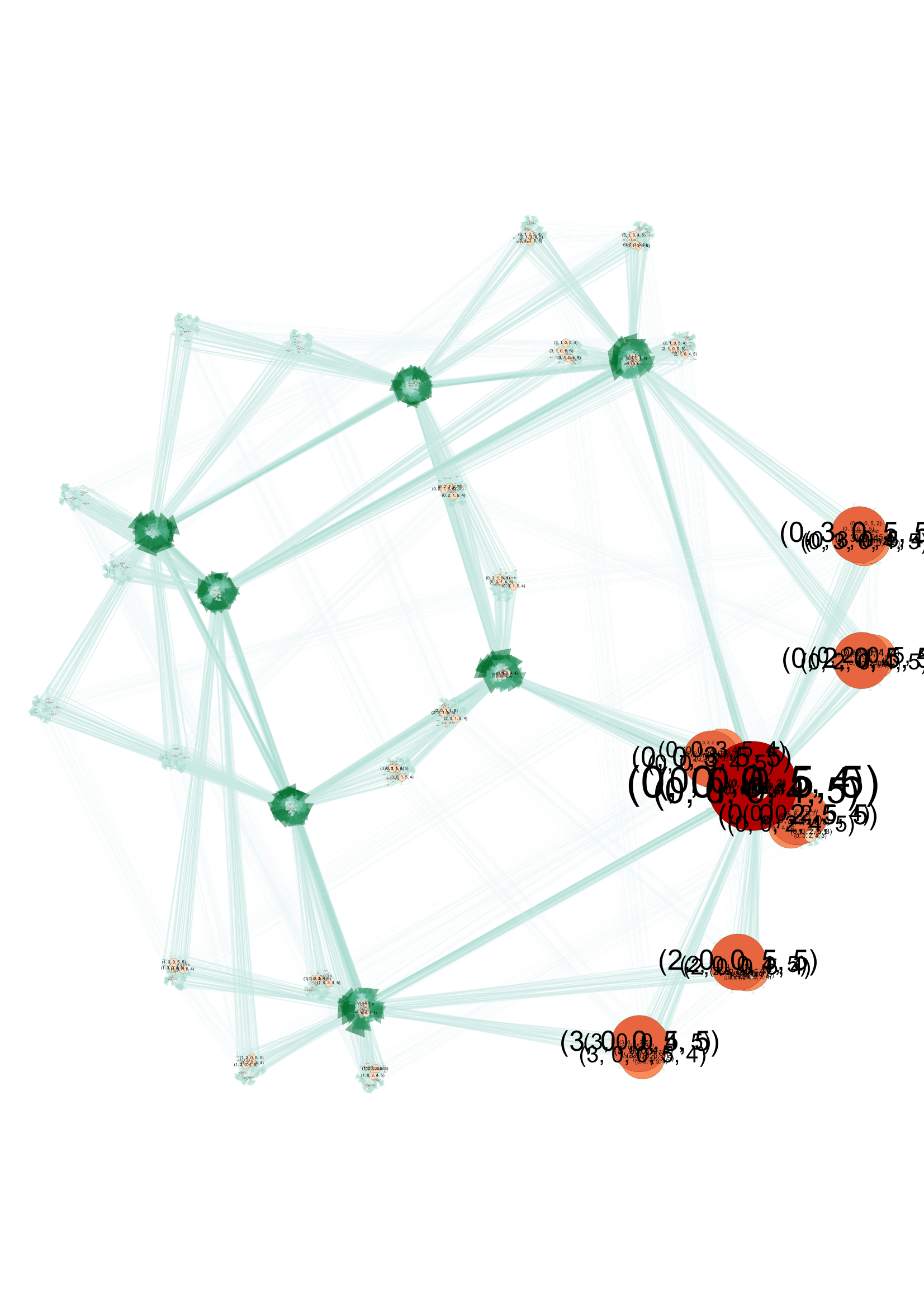}}
	\caption{Graph visualizations with a linear-repulsion linear-attraction model of the aligned environment's game graph under Boston with different varying values for $v_S$, and fixed values $v_G=100$, $v_B=25$ for all students. The nodes' sizes and colors represent the $\alpha$-Rank distribution's mass. The edges are green, and their brightness and size depict the improvement in the deviating player's utility.}
	\label{fig:graph-vizs-boston}
\end{figure*}

Finally, we analyze the aligned environment under Boston as a Bayesian game, which we denote by $\text{AE}$, using $\alpha$-Rank collections for some given prior distribution. 
Suppose that there was an effort to estimate the players' vNM functions with a conclusion that the Gold school is ranked $v_{G, i} \approx 100$, Silver with $v_{S, i} \approx 70$, and Bronze $v_{B, i} \approx 25$ for every $i \in \mathcal{N}$.
If the estimation is fairly precise, one can expect that the corresponding $\alpha$-Rank-collection is close to the results depicted in Figure~\ref{fig:graph-viz-standard-case-boston}. With increasing uncertainty, more weight may fall into strategically different cases, such as those shown in Figures~\ref{fig:graph-viz-indifferent-case-boston} and \ref{fig:graph-viz-strictly-case-boston}. 
We assume that the uncertainty about these values is captured by normal distributions $\mathcal{N}(\mu, \sigma)$, where we assume that larger values have a higher standard deviation $\sigma$. More specifically, the prior distribution $F$ is given such that $v_{G, i} \sim \mathcal{N}(100, 6)$, $v_{S, i} \sim \mathcal{N}(70, 3)$, and $v_{B, i} \sim \mathcal{N}(25, 2)$. The value for being unmatched is fixed to be $v_{U, i} = 0$. There is a probability that a drawn vNM function does not satisfy the ordinal preference relation. If this is the case, the types are resampled. The strength of the $\alpha$-Rank-collection concept is that it captures the variance in strategic play that emerge from uncertain utilities. A high uncertainty likely results in a wider variety of relevant strategies, whereas a low uncertainty likely allows a more precise prediction.

We approximate the $\alpha$-Rank-collection $\Lambda_{\text{AE}}$ (see Equation~\ref{equ:approx-alpha-rank-collection}) via Monte-Carlo integration. We sample $1000$ different types, calculate the corresponding $\alpha$-Rank distributions and average the results. We chose the $\alpha$ in the following way. We determined the largest possible $\alpha$ so that the Markov chain is still connected for twenty different vNM functions around the distributions' mean values. In all cases, we found $\alpha_{u(\cdot, v)}^{\text{pre}} = 6.71$. We use this value for all simulations unless the Markov chain is not connected. In that case, we keep decreasing $\alpha$ by $0.1$ until a unique stationary distribution can be calculated. The smallest used $\alpha$ in the simulations is $5.81$.  

The approximated $\alpha$-Rank-collection $\tilde{\Lambda}_{\text{AE}}$ yields the following results.
Table~\ref{tab:strat-dist-alpha-rank-collection-matching} shows the average mass on strategies per student group. The mass on Tops playing truthfully is about $20.5$\%, on their equilibrium strategy its close to $70$\%, but roughly $10$\% on Silver first. The mass on Averages playing one of their equilibrium strategies is also about $70$\%, whereas Bronze first is $28$\% and Gold first about $1$\%. However, when considering the joint strategies played, only about $60$\% of the expected strategy strength lies on the four strategies in $\text{NE}_{\text{Bo}}$, even though it would be close to $100$\% for the prior's mean values.

\begin{table}
	\centering
	\caption{Average strategy distribution among student groups according to $\tilde{\Lambda}_{\text{AE}}$.}
	\label{tab:strat-dist-alpha-rank-collection-matching}
	\begin{tabular}{rcccccc}
		\toprule
		Strategy & (G, S, B) & (G, B, S) & (S, G, B) & (S, B, G) & (B, G, S) & (B, S, G) \\
		\midrule
		Top     & 0.205  & 0.698  & 0.048  & 0.048 & 0.000  & 0.000  \\
		Average & 0.007     & 0.004     & 0.274     & 0.438     & 0.105     & 0.172   \\ 
		\bottomrule
	\end{tabular}
\end{table}

The expected outcome distribution weighted by the mass given by $\tilde{\Lambda}_{\text{AE}}$ for a Top and Average is shown in Table~\ref{tab:outcome-dist-alpha-rank-collection-matching}. 
A Top gets a Gold seat about $\left( \frac{2}{3} \right)$ of the times, meeting the equilibrium prediction for any strategy in $\text{NE}_{\text{Bo}}$. However, she also expects to receive Silver about $11.2$\% and Unmatched $1.6$\% of the time. In total, the mass put on strategies wherein the expected outcomes one Top remains unmatched is $4.7$\%, one Average gets the Silver seat is $66.5$\%, and one Average gets the Bronze seat is $38.2$\%. Each of these outcomes describes a somewhat surprising scenario given the strategic interaction. For example, one might suspect that no Top would remain unmatched due to her being higher ranked by the schools than the Averages.

\begin{table}
	\centering
	\caption{Expected outcome distribution for students by group according to $\tilde{\Lambda}_{\text{AE}}$. The equilibrium prediction for the mean values follow in parentheses.}
	\label{tab:outcome-dist-alpha-rank-collection-matching}
	\begin{tabular}{rcccc}
		\toprule
		{} &   Gold &  Silver &  Bronze &  Unmatched \\
		\midrule
		Top     &  0.667 $\left(\frac{2}{3}\right)$ &   0.112 $\left(0\right) $&   0.206 $\left(\frac{1}{3}\right)$ &      0.016 $\left(0\right) $ \\
		Average &  0.000 $\left(0\right) $ &   0.333 $\left(\frac{1}{2}\right)$ &   0.191 $\left(0\right) $ &      0.476 $\left(\frac{1}{2}\right)$ \\
		\bottomrule
	\end{tabular}
\end{table}

In summary, the analysis of the matching market shows that the $\alpha$-Rank-collection provides a more nuanced strategic analysis than equilibrium predictions for a specific vNM function, evaluating the strength for range of strategy profiles. Furthermore, while $\alpha$-Rank is sensitive to specific vNM functions, $\alpha$-Rank-collections take their inherent uncertainty into account. Finally, note that the analytical determination of the BNE for this example is highly non-trivial due to the complex prior distributions and might change drastically with a slightly different prior, rendering the previous results obsolete. In contrast, $\alpha$-Rank-collections can be easily adapted to modeling changes, and previously calculated $\alpha$-Rank distributions can be reused in follow-up calculations. Overall, $\alpha$-Rank-collections provide a principled approach to analyzing strategies in uncertain situations in players' vNM functions, where equilibrium solution concepts fall short.

\section{Conclusion and Future Work}
In many decision problems, specific vNM functions are not available; not even the players know their vNM function with certainty. This can be particularly challenging in matching markets where non-strategyproof mechanisms are used, as the intensity of preferences over outcomes can greatly impact strategic play. However, until now, there has been a lack of formal approaches to analyzing these types of games strategically.
We use Bayesian games to model uncertainty in players' vNM functions via a prior distribution. This achieves two things. First, the modeling approach offers the option to enrich ordinal, or generally unspecific, preferences by uncertain information about the strength or intensity of preferences that agents typically have over alternatives. Second, it shifts the focus of analysis towards determining the expected strategic behavior for a collection of normal-form games, where players know their own type and those of all others in the interim stage, and the game-theoretic prediction needs to consider strategic play given the ex-ante uncertainty.

Previous approaches that aim to solve Bayesian games, such as the Bayesian Nash equilibrium, rationalize strategies by considering the whole opponents' prior distribution. Each player knows her type in the interim stage and rationalizes over other players' types and strategies. In turn, they are collapsing strategic play into a fixed point in the ex-ante strategy space, which suffers from high sensitivity to modeling assumptions and intractability.
Instead, we propose analyzing players' behavior for each game instance separately, weighting it according to the prior distribution to obtain an expected strategic play that captures uncertainty in behavior across different games. We use $\alpha$-Rank to capture players' long-term behavior in a single normal-form game, providing a unique characterization of strategies that emerge from finite population replicator dynamics that represents how likely agents are in a particular strategy profile. Consequently, we introduce $\alpha$-Rank-collections, a solution concept that combines the agents' distributional information about the intensity of preferences with the expressive prowess of $\alpha$-Rank. Importantly, it satisfies several desirable properties. For example, we show that it is invariant to positive affine transformations so that it is independent of the choice of vNM functions' representatives and is efficient to approximate. Furthermore, we demonstrate in two examples that it provides a more nuanced analysis of strategic behavior than equilibrium predictions. In contrast to ordinal Nash equilibria, existence is guaranteed.

A limitation of $\alpha$-Rank is that the game graph grows exponentially in the number of players and strategies, limiting the application to relatively small games. Already our example of the Boston mechanism with five participants and six strategies led to a matrix with roughly $40 \ 000$ entries. Recent work is addressing this problem \citep{yangAlphaAlphaRank2020}, but scalability is also a topic for future research. We want to study larger game instances by applying more sophisticated methods than simple Monte-Carlo approximation. Besides, we want to analyze different types of matching markets and the predictive accuracy of $\alpha$-Rank-collections.

\bibliographystyle{unsrtnat}
\bibliography{references}  






\appendix

\section{Proof of Invariance of \texorpdfstring{$\alpha$}{alpha}-Rank-Collections to Positive Affine Transformations} \label{sec:proof-invariance-positive-affine-transformations}

This section provides the proof for Theorem~\ref{thm:alpha-rank-collection-invariant-to-positive-affine-transformations}.
\begin{proof}
	\ For the first statement, consider the Bayesian game $\text{BG} = (\mathcal{N}, \mathcal{S}, \mathcal{V}, F, u)$ and the values $a > 0$ and $b \in \mathbb{R}$. When referring to a scalar variable in boldface, we assume it to be a vector with equal values as the scalar of appropriate size, given the context, e.g., $\mathbf{b} = (b, b, \dots, b)^T$. Define the Bayesian game $\overline{\text{BG}} = (\mathcal{N}, \mathcal{S}, \overline{\mathcal{V}}, F, u)$ with $\overline{\mathcal{V}} = \{av + \mathbf{b} | v \in \mathcal{V}\}$. Then we need to show that $\Lambda_{\text{BG}} = \Lambda_{\overline{\text{BG}}}$. These $\alpha$-Rank collections are an expectation over the types for the limiting $\alpha$-Rank distribution. Therefore, it is sufficient to show that for any $v \in  \mathcal{V}$ with $\overline{v} = av + b$, it holds that $g(u(\cdot, v)) = g(u(\cdot, \overline{v}))$.
	
	Let $s, s^{\prime} \in \mathcal{S}$ be two strategy profiles with a single deviating player $i$. Then for $v \in \mathcal{V}$ with $u_i((s^{\prime}_i, s_{-i}), v) \neq u_i((s_i, s_{-i}), v)$ the transition probability from $s$ to $s^{\prime}$ satisfies
	\begin{align*}
		C_{s, s^{\prime}}(\alpha, u(\cdot, \overline{v})) 
		&= \eta \frac{1 - e^{- \alpha (u_i((s^{\prime}_i, s_{-i}), \overline{v}) - u_i((s_i, s_{-i}), \overline{v}))}}{1 - e^{- m\alpha (u_i((s^{\prime}_i, s_{-i}), \overline{v}) - u_i((s_i, s_{-i}), \overline{v}))}} \\
		&= \eta \frac{1 - e^{- \alpha (u_i((s^{\prime}_i, s_{-i}), a v + \mathbf{b}) - u_i((s_i, s_{-i}), a v + \mathbf{b}))}}{1 - e^{- m\alpha (u_i((s^{\prime}_i, s_{-i}), a v + \mathbf{b}) - u_i((s_i, s_{-i}), a v + \mathbf{b}))}} \\
		&= \eta \frac{1 - e^{- a \cdot \alpha (u_i((s^{\prime}_i, s_{-i}), v) - u_i((s_i, s_{-i}), v))}}{1 - e^{- a \cdot m\alpha (u_i((s^{\prime}_i, s_{-i}), v) - u_i((s_i, s_{-i}), v))}} \\
		&= C_{s, s^{\prime}}(a \cdot \alpha, u(\cdot, v))
	\end{align*}
	As $a > 0$, it holds that
	\begin{align*}
		\lim_{\alpha \rightarrow \infty} C_{s, s^{\prime}}(\alpha, u(\cdot, \overline{v}))  &= \lim_{\alpha \rightarrow \infty} C_{s, s^{\prime}}(a \cdot \alpha, u(\cdot, v)) \\
		&= \lim_{\alpha \rightarrow \infty} C_{s, s^{\prime}}(\alpha, u(\cdot, v)) .
	\end{align*}
	The transition probabilities of self-connections are given by Equation \ref{equ:alpha-rank-transition-prob-self-transition} and are a finite sum of non-self-connections. Consequently, the probability's limit of $\alpha \rightarrow \infty$ is again equal for $\overline{v}$ and $v$. All remaining transition probabilities are independent of the types. 
	Therefore, we get
	\begin{align*}
		g(u(\cdot, \overline{v})) &= \lim_{\alpha \rightarrow \infty} f( \alpha, u(\cdot, \overline{v})) = \lim_{\alpha \rightarrow \infty} f(\alpha, u(\cdot, a v + \mathbf{b})) \\
		&= \lim_{\alpha \rightarrow \infty} f(a \cdot \alpha, u(\cdot, v)) = \lim_{\alpha \rightarrow \infty} f(\alpha, u(\cdot, v)) = g(u(\cdot, v)),
	\end{align*}
	which shows the first statement.
	
	For the second part, we consider the setting with private values, i.e., $u_i(s, v) = u_i(s, v_i)$ for all $s \in \mathcal{S}$, $v \in \mathcal{V}$, and $i \in \mathcal{N}$. Let $a_i > 0$ and $b_i \in \mathbb{R}$ and denote $\overline{v}_i = a_i v_i + b_i$ for $v = (v_i, v_{-i})$ with $v_i \in \mathcal{V}_i$ and $v_{-i} \in \mathcal{V}_{-i}$. Then it holds that $u_i(s, a_i v + \mathbf{b}_i) = u_i(s, a_i v_i + \mathbf{b}_i)$. Repeating the steps above gives us $g(u(\cdot, v)) = g(u(\cdot, (v_{-i}, \overline{v}_i)))$. As this holds for every $i \in \mathcal{N}$, the second statement follows completing the proof. 
\end{proof}

\section{Proof of Continuity of $\alpha$-Rank} \label{sec:proof-of-continuity-of-alpha-rank}
In this section, we present the proof of Theorem \ref{thm:continuity-of-alpha-rank-in-vNM-and-alpha}. Before establishing this, we present some results from linear algebra as an intermediate step. We start by defining primitive matrices and their connection to irreducibility and aperiodicity. Consequently, we state the famous Perron-Frobenius theorem, which was first stated in this form by \citet{perron-frobenius-theorem-original-source}. The following definition and two theorems are taken from \citet{Seneta2006}.

\begin{definition}
	A square non-negative matrix $B$ is said to be primitive if there exists a positive integer $k$ such that $B^k > 0$. 
\end{definition}

\begin{theorem} \label{thm:irreducible-and-aperiodic-matrix-equals-primitive}
	A matrix $B$ irreducible and aperiodic if and only if it is primitive.
\end{theorem}

The Perron-Frobenius theorem for primitive matrices is a follows.

\begin{theorem} \label{thm:perron-frobenius}
	Suppose $B\in \mathbb{R}^{n \times n}$ is a non-negative primitive matrix. Then there exists an eigenvalue $r$ such that:
	\begin{description}
		\item [\namedlabel{thm:perron-frobenius-real-larger-0}{(a)}] $r$ is a real value and strictly larger than $0$ 
		\item [\namedlabel{thm:perron-frobenius-positive-left-right-eigenvectors}{(b)}] with $r$ can be associated strictly positive left and right eigenvectors  
		\item [\namedlabel{thm:perron-frobenius-r-largest-eigenvalue}{(c)}] $r > \norm{\lambda}$ for any eigenvalue $\lambda \neq r$ 
		\item [\namedlabel{thm:perron-frobenius-eigenvectors-unique-up-to-scalar}{(d)}] the eigenvectors associated with $r$ are unique to constant multiples  
		\item [\namedlabel{thm:perron-frobenius-r-simple-root}{(e)}] $r$ is a simple root of the characteristic polynomial of $B$ 
	\end{description}
	One calls $r$ the Perron-Frobenius eigenvalue and its corresponding positive eigenvectors, the Perron-Frobenius eigenvectors.
\end{theorem}

The following theorem gives conditions so that eigenvectors of simple eigenvalues of a matrix-valued function are continuous in changes of the matrix's entries. We state an adapted version of Theorem 8 on page 130 by \citet{lax2007linear}.

\begin{theorem} \label{thm:continuity_eigenvalue_eigenvector_by_lax}
	Let $B(t)$ be a square matrix whose elements are continuously differentiable in $t \in \mathbb{R}^m$. Suppose that $r_0$ is an eigenvalue of $A(0)$ of multiplicity one, in the sense that $r_0$ is a simple root of the characteristic polynomial of $B(0)$. Then there exists a $\delta > 0$ such that for $\norm{t} \leq \delta$, there exists an eigenvalue $r(t)$ of $B(t)$ that depends continuously differentiable on $t$, with $r(0) = r_0$. Furthermore, we can choose an eigenvector $h(t)$ of $B(t)$ pertaining to the eigenvalue $r(t)$ to depend continuously differentiable on $t$.
\end{theorem}

The version stated by \citeauthor{lax2007linear} considered the case $m=1$. Extending the result to the $m > 1$ case additionally demands that the derivatives of $B(t)$ are continuous in $t \in \mathbb{R}^m$. Consequently, the proof carries over without further changes. Using the previous results, we show a lemma that establishes the continuity of stationary distributions.

\begin{lemma} \label{thm:continuity_stationary_distribution_in_policy}
	Let the function $B: \mathbb{R}^m \rightarrow \mathbb{R}^{n \times n}$ be continuously differentiable. Furthermore, for every $t \in \mathbb{R}^m$, the resulting matrix $B(t)$ is a transition matrix of an irreducible and aperiodic Markov chain.
	Then the mapping on the stationary distribution given by $\pi: \mathbb{R}^m \rightarrow \mathbb{R}^n$ is continuously differentiable. 
\end{lemma}

\begin{proof}
	\ The stationary distribution $\pi(t)$ to the matrix $B(t)$ exists and is unique, as $B(t)$ is irreducible and aperiodic (see Theorem 4.1 on page 119 \citep{Seneta2006}), which establishes that the function $\pi$ is well-defined. 
	
	To show that $\pi$ is continuously differentiable, we aim to use Theorem \ref{thm:continuity_eigenvalue_eigenvector_by_lax} for the matrix $B(t)^T$ as the stationary distribution $\pi(t)$ is a right eigenvector of $B(t)^T$ to the eigenvalue one.
	
	Note that by Theorem \ref{thm:irreducible-and-aperiodic-matrix-equals-primitive}, the matrix $B(t)$ is primitive, as it is an irreducible and aperiodic square matrix. That means there exists an $k \in \mathbb{N}$ such that $B(t)^k > 0$, showing that Theorem \ref{thm:perron-frobenius} holds for $B(t)$. The matrix $B(t)$ is row-stochastic so that its largest eigenvalue is one, which means by part \ref{thm:perron-frobenius-r-largest-eigenvalue} of Theorem \ref{thm:perron-frobenius} that one is also the Perron-Frobenius eigenvalue.
	Note that
	\begin{align*}
		B(t)^k > 0 \Leftrightarrow \left(B(t)^k \right)^T  > 0 \Leftrightarrow \left( B(t)^T \right)^k > 0.
	\end{align*}
	That means, $B(t)^T$ is primitive, so that the Perron-Frobenius theorem holds for $B(t)^T$ as well. Furthermore, for $\lambda \in \mathbb{R}$ it holds that
	\begin{align*}
		\det \left(B(t)^T - \lambda I \right) = \det \left( \left( B(t) - \lambda I \right)^T \right) = \det \left(B(t) - \lambda I \right),
	\end{align*}
	which means that $B(t)^T$ has the same eigenvalues as $B(t)$. Therefore, the eigenvalue one is also the Perron-Frobenius eigenvalue of $B(t)^T$.
	
	By part \ref{thm:perron-frobenius-r-simple-root} of Theorem \ref{thm:perron-frobenius}, one is a simple root of $B(t)^T$'s characteristic polynomial.
	We can now use Theorem \ref{thm:continuity_eigenvalue_eigenvector_by_lax} for $B(t)^T$. For this, let $t_0 \in \mathbb{R}^m$ be arbitrary. 
	Then, there exists a $\delta > 0$ such that, on the set $D_{t_0}:= \left\{ t \in \mathbb{R}^m: \norm{t - t_0} \right\}$, there are continuously differentiable functions $r: D_{t_0} \rightarrow \mathbb{R}, h: \Lambda \rightarrow \mathbb{R}^n$. Where $r(t)$ is an eigenvalue of $B(t)^T$ with $r(t_0) = 1$ and $h(t)$ is an eigenvector of $B(t)^T$ of the eigenvalue $r(t)$.
	
	By part \ref{thm:perron-frobenius-r-largest-eigenvalue} of Theorem \ref{thm:perron-frobenius}, all eigenvalues $\lambda$ of $B(t)^T$ with $\lambda \neq 1$ satisfy $\norm{\lambda} < 1$ for all $t \in D_{t_0}$. Therefore, it follows that $r \equiv 1$, as it is continuously differentiable on $D_{t_0}$. Therefore, $h(t)$ is a strictly positive eigenvector for all $t \in D_{t_0}$ according to part \ref{thm:perron-frobenius-positive-left-right-eigenvectors}.
	We define the scaling function
	\begin{align*}
		p(t) := \left( \sum_{i = 1}^{n} h(t)_{i} \right)^{-1} \Leftrightarrow p(t) \cdot \left( \sum_{i = 1}^{n} h(t)_{i} \right) = 1.
	\end{align*}
	Consequently, for all $t \in D_{t_0}$, the vector $p(t) \cdot h(t)$ is an eigenvector of $B(t)^T$ to the eigenvalue one. Particularly, $p(t) \cdot h(t)$ is a stationary distribution of $B(t)$, which means $\pi(t) = p(t)h(t)$ for all $t \in D_{t_0}$. Note that $p$ is continuously differentiable as $h$ is a strictly positive continuously differentiable function. Therefore, $\pi$ is continuously differentiable for all $t \in D_{t_0}$. As $t_0 \in \mathbb{R}^m$ has been arbitrary, $\pi$ is continuously differentiable for all $t \in \mathbb{R}^m$, which finishes the proof.
\end{proof}

Finally, before showing the main result, we need the following auxiliary lemma.

\begin{lemma} \label{thm:continuously_extendable_function_at_0}
	Let $m \in \mathbb{N}$. Then the function
	\begin{align} \label{equ:auxiliary-exp-fraction-function}
		h(x) = \begin{cases}
			\frac{m - 1 + e^{mx} - m e^{x}}{\left(e^{mx} - 1 \right)^2}  & \text{, } x \neq 0\\
			\frac{m-1}{2m} & \text{, }x = 0
		\end{cases}
	\end{align}
	is continuous for $x \in \mathbb{R}$.
\end{lemma} 

\begin{proof}
	\ For $x \neq 0$, the function is continuous as concatenation of continuous functions. It remains to be shown that $h$ is continuous in $x = 0$. By using the series expansion of the exponential function given by $e^x = \sum_{n=0}^{\infty} \frac{x^n}{n!}$, we can write $h$ for $x \neq 0$ as
	\begin{align*}
		h(x) =& \frac{m - 1 + \sum_{n=0}^{\infty} \frac{-m x^n + (mx)^n}{n!}}{1 + \sum_{n=0}^{\infty} \frac{(2 m x)^n}{n!} - 2 \cdot \sum_{n=0}^{\infty} \frac{(mx)^n}{n!}} 
		= \frac{\frac{1}{2} \left(-mx^2 + m^2x^2 \right) + \mathcal{O}(x^3)}{m^2 x^2 + \mathcal{O}(x^3)} \\
		=& \frac{m -1 + \mathcal{O}(x)}{2m + \mathcal{O}(x)},
	\end{align*}
	where $\mathcal{O}$ denotes the Landau symbol. Therefore, we get $\lim_{x \rightarrow 0} h(x) = \frac{m - 1}{2m}$ which completes the proof.
\end{proof}

With the previous results established, we can finally proof Theorem \ref{thm:continuity-of-alpha-rank-in-vNM-and-alpha}.

\begin{proof}
	\ Denote with $C: (0, \infty) \times \mathbb{R}^{\sum_i m_i} \rightarrow \mathbb{R}^{|\mathcal{S}| \times |\mathcal{S}|}, (\alpha, v) \mapsto C(\alpha, u(\cdot, v))$ the matrix-valued function that maps the tuple of ranking intensity and vNM functions to the Markov-chain's transition matrix with entries $C_{s, s^{\prime}}(\alpha, u(\cdot, v))$ as defined in Section \ref{subsec:alpha-rank}. Then by Theorem 2.2.1 of \citep{omidshafieiARankMultiAgentEvaluation2019}, it is a transition matrix of an irreducible and aperiodic Markov-chain for every $\alpha$ and $v$. To use Lemma \ref{thm:continuity_stationary_distribution_in_policy} and show the statement, it remains to be shown that $C(\alpha, u(\cdot, v))$ is a continuously differentiable function. 
	
	Denote with $a^{i}_{s, s^{\prime}}(v_i) := u_i((s_i^{\prime}, s_{-i}), v_i) - u_i((s_i, s_{-i}), v_i)$ the difference for a unilaterally deviation of player $i$ and define $h: \mathbb{R} \rightarrow \mathbb{R}$ as in Equation \ref{equ:auxiliary-exp-fraction-function}. The function $a^{i}_{s, s^{\prime}}$ is continuously differentiable as $u_i$ is continuously differentiable in its second argument by assumption. Let $s, s^{\prime} \in \mathcal{S}, v_i \in \mathcal{V}_i$ be such that $a^{i}_{s, s^{\prime}}(v_i) \neq 0$. Consider the partial derivatives
	\begin{align}
		\frac{d}{d \alpha}C_{s, s^{\prime}}(\alpha, u(\cdot, v)) &= \eta \cdot e^{\alpha a^{i}_{s, s^{\prime}}(v_i) (m-1)} a^{i}_{s, s^{\prime}}(v_i) \cdot \frac{m -1 + e^{\alpha m a^{i}_{s, s^{\prime}}(v_i)} - m e^{\alpha a^{i}_{s, s^{\prime}}(v_i)}}{\left(e^{\alpha m a^{i}_{s, s^{\prime}}(v_i)} - 1 \right)^2} \nonumber \\
		&= \eta \cdot e^{\alpha a^{i}_{s, s^{\prime}}(v_i) (m-1)} a^{i}_{s, s^{\prime}}(v_i) \cdot h\left(\alpha a^{i}_{s, s^{\prime}}(v_i) \right) \label{equ:partial-derivatives-continuity-proof-d-alpha}\\
		\frac{d}{d v_i}C_{s, s^{\prime}}(\alpha, u(\cdot, v)) &= \eta \cdot \alpha \cdot e^{\alpha a^{i}_{s, s^{\prime}}(v_i) (m-1)} a^{i}_{s, s^{\prime}}(v_i) \cdot h\left(\alpha a^{i}_{s, s^{\prime}}(v_i) \right) \label{equ:partial-derivatives-continuity-proof-d-v}.
	\end{align}
	By Lemma \ref{thm:continuously_extendable_function_at_0}, the function $h$ is continuous and as the mapping $(x, y) \mapsto x \cdot y$ is continuous, $h\left(\alpha a^{i}_{s, s^{\prime}}(v_i) \right)$ is continuous in $\alpha$ and $v_i$ for every $i \in \mathcal{N}$. Therefore, the partial derivatives of Equations \ref{equ:partial-derivatives-continuity-proof-d-alpha} and \ref{equ:partial-derivatives-continuity-proof-d-v} are continuous as product of continuous functions. The entries representing self-connections, $C_{s, s}(\alpha, u(\cdot, v))$ are a sum of continuously differentiable functions and, therefore, also continuously differentiable. That means, the matrix valued function $C(\alpha, u(\cdot, v))$ is continuously differentiable in $\alpha$ and $v$. By Lemma \ref{thm:continuity_stationary_distribution_in_policy}, we get that the mapping on the stationary distribution $f(\alpha, u(\cdot, v))$ is continuous in $\alpha$ and $v$, which finishes the proof. 	
\end{proof}

\end{document}